\documentclass[10pt]{article}

\usepackage[T1]{fontenc}
\usepackage{lmodern}
\usepackage[utf8]{inputenc} 
\usepackage{hyperref}       
\usepackage{url}            
\usepackage{booktabs}       
\usepackage{amsfonts}       
\usepackage{nicefrac}       
\usepackage{microtype}      
\usepackage{graphicx}
\usepackage{natbib}
\usepackage{doi}

\usepackage{authblk}

\usepackage{amsmath}
\usepackage{tikz}
\usepackage{makecell}
\usepackage{caption}
\usepackage{appendix}
\usepackage{tablefootnote}

\usepackage{amsthm}                                                                        
                                                                                             
  \newtheorem{theorem}{Theorem}[section]                                                     
  \newtheorem{lemma}[theorem]{Lemma}
  \newtheorem{proposition}[theorem]{Proposition}  

\title{Unity and Diversity of Intracellular pH Maintenance Mechanisms}
\author[1]{Guillaume Terradot}
\author[2]{Vincent Danos}
\affil[1]{Independent Researcher}
\affil[2]{Centre National de la Recherche Scientifique, 75794 Paris, France; and École Normale Supérieure, 75005 Paris, France}
\date{}

\begin{document}

\maketitle

\begin{abstract}
All cells must sustain ionic motive forces (IMFs)---the
electrochemical gradients of permeant ions, together with the
membrane potential they produce---to regulate intracellular pH,
drive secondary transport, and power ATP synthesis. Because membranes
are imperfectly impermeable, IMFs continuously dissipate through
passive leakage, and active transport must compensate at an energetic
cost that competes with growth and biosynthesis. How environmental
conditions set this cost, and why cells across the tree of life
share a common ionic logic yet deploy strikingly diverse transporter
repertoires, has lacked a unifying quantitative account.

Here we derive a thermodynamic lower bound on the power required to
maintain IMFs at steady state. The bound equals the rate of free-energy
dissipation by ion leakage, holds across a broad family of
electrophysiological models, and is independent of organism, energy
source, or transporter architecture. Using this framework, we quantify
how extracellular pH, salinity, and temperature reshape the energetic
burden of electrophysiological homeostasis.

Cost minimization recovers, from first principles, the universal
K\textsuperscript{+}-rich, Na\textsuperscript{+}-poor cytoplasm observed
across taxa: asymmetric membrane permeabilities alone are sufficient to
explain it. The same framework predicts that extremophiles face
systematically higher maintenance costs under extreme pH, salinity, and
temperature, and that when sustaining a large proton motive force
becomes prohibitive, cells should shift to metabolic regimes compatible
with smaller PMF, providing a thermodynamic rationale for
stress-induced metabolic reconfiguration.

Finally, we show that perfect energetic efficiency is unattainable in
practice, and that this very imperfection, combined with environmental
variability, selects for the diversity of transport architectures
observed in nature: each architecture is optimal within a discrete
regime of environmental constraints. Together, these results turn IMF
maintenance from a qualitative housekeeping cost into a predictive
thermodynamic quantity, linking environmental physics to the unity and
diversity of cellular ion transport.
\end{abstract}

\section{Introduction}

All cells maintain ionic motive forces (IMFs)---the electrochemical
gradients of permeant ions across the plasma membrane, together with
the membrane potential ($\Delta\psi$) they entail. Two of these
gradients are physiologically critical: a near-neutral intracellular
pH ($\mathrm{pH}_i$), required for protein stability and biochemical
activity~\cite{yang1993,krulwich2011,Aoi2014}, and a sufficiently negative proton
motive force (PMF), required for respiratory catabolism, ATP synthesis
via $\mathrm{F_1F_o}$, and secondary transport~\cite{Yan2013}. Because
$\mathrm{pH}_i$, PMF, and $\Delta\psi$ are coupled through the same
set of permeant ions, they form a single electrophysiological system
that cells must regulate jointly despite fluctuations in extracellular
pH, salinity, and temperature.

Maintaining this system has an unavoidable thermodynamic cost.
Membranes are imperfectly impermeable, so IMFs continuously dissipate
through passive ion leakage, and active transport must compensate at
steady state. The resulting energetic burden competes directly with
growth and biosynthesis. Classical pump--leak models capture this
balance mechanistically~\cite{tosteson1960,Kay2017}, linking transport fluxes to
permeabilities and electrochemical gradients, but they remain
organism- and architecture-specific.

\emph{No general framework currently sets a universal lower bound on
the energetic cost of sustaining IMFs, nor predicts how that cost
varies with environmental conditions.}

This gap matters because environmental variation imposes very
different electrophysiological challenges: acidic environments elevate
proton influx, alkaline conditions demand strongly negative
$\Delta\psi$, high salinity increases cation leakage, and elevated
temperature raises membrane permeability. Whether the costs imposed by
these conditions follow predictable thermodynamic rules---and whether
deeply conserved features such as the K\textsuperscript{+}-rich,
Na\textsuperscript{+}-poor cytoplasm fall out of dissipation
minimization---has remained unresolved. So has the converse puzzle:
why, on a shared ionic logic, cells deploy such diverse transporter
repertoires.

Here we derive a thermodynamic lower bound on the power required to
maintain IMFs at steady state. The bound holds across a broad family
of electrophysiological models and equals the rate of free-energy
dissipation by passive ion leakage; it is independent of organism,
energy source, and the molecular details of active transport. Using
this framework, we then quantify how extracellular pH, salinity, and
temperature reshape the maintenance burden of $\mathrm{pH}_i$, PMF,
and $\Delta\psi$ regulation.

The framework yields three classes of prediction. First, cost
minimization recovers the universal K\textsuperscript{+}-rich,
Na\textsuperscript{+}-poor cytoplasm from asymmetric membrane
permeabilities alone. Second, it constrains transporter stoichiometry
and predicts shifts between catabolic regimes when sustaining a large
PMF becomes prohibitive, providing a thermodynamic rationale for
stress-induced metabolic reconfiguration. Third, it explains why
perfect energetic efficiency is unattainable in practice and why this
imperfection, combined with environmental variability, selects for the
diversity of transport architectures observed in nature---each
optimal within a discrete regime of environmental constraints.

The remainder of the paper is organized as follows. Section~\ref{sec:2.1} derives
the thermodynamic bound and validates it on two canonical transport
architectures. Section~\ref{sec:2.3} applies it to the prediction of intracellular
composition, section~\ref{sec:2.4} to environmental gradients of pH, salinity, and temperature section~\ref{sec:2.5} to different catabolic regimes. Section~\ref{sec:2.6} and \ref{sec:2.7} turns to the unity-and-diversity question,
showing how environmental variability partitions transporter
architectures into distinct optimal regimes. The appendix collects modelling details and simulations.

\section{Results}

\subsection{A thermodynamic bound on the cost of maintaining ionic motive forces}
\label{sec:2.1}

The central result of this section is a universal lower bound on the
power cells must spend to sustain ionic motive forces (IMFs): at
steady state, the active power $\mathcal{P}_{\mathrm{active}}$
delivered by transporters cannot fall below the rate
$\mathcal{P}_{\mathrm{leak}}$ at which free energy is dissipated by
passive ion leakage,
\begin{equation}
  \mathcal{P}_{\mathrm{active}} \geq \mathcal{P}_{\mathrm{leak}}.
  \label{eq:bound-informal}
\end{equation}
The bound follows from thermodynamics alone: at steady state, every
unit of free energy lost to leakage must be replaced by an equal
amount of work done by active transport, and any additional
dissipation in the transport reactions themselves can only widen the
gap. It is therefore independent of organism, energy source, and the
molecular details of active transport. In particular, it applies
identically to the two transporter architectures we will use as
running examples (Fig.~1).

In what follows we use ``energetic cost,'' ``power,'' and
``dissipation'' interchangeably for steady-state free-energy
expenditure per unit time, and we identify active expenditure with
leakage dissipation whenever the bound is saturated.

\paragraph{Kinetic framework.}
To turn~(\ref{eq:bound-informal}) into a quantitative prediction we
need explicit expressions for the leakage flux of each permeant ion
and for the free energy it dissipates. Throughout, we measure
free-energy changes in volts: $\Delta G_r$ denotes the reaction
free energy per unit charge, so that $F\,\Delta G_r$ is the
free-energy change in J/mol. We adopt a thermodynamically consistent
rate law~\cite{beard2008}: for a generic reaction $r$ with free-energy
change $\Delta G_r$, the net flux is
\begin{equation}
  j_r \;=\; j_r^{+}\!\left(1 - e^{\eta\, \Delta G_r}\right),
  \qquad \eta \equiv \frac{F}{RT},
  \label{eq:flux-generic}
\end{equation}
where $j_r^{+}$ is the unidirectional forward flux---the rate at
which $r$ proceeds in the forward direction in the absence of the
reverse reaction. Equation~(\ref{eq:flux-generic}) vanishes at
equilibrium ($\Delta G_r = 0$) and changes sign with $\Delta G_r$,
as required by the second law.

For the passive leakage of an ion $x$ across the membrane, the only
reaction is $x_{e} \rightleftharpoons x_{i}$, whose free-energy change
(in volts) is exactly the ionic motive force $\Delta G_x$. Specializing
(\ref{eq:flux-generic}) gives
\begin{equation}
  j_x \;=\; j_x^{+}\!\left(1 - e^{\eta\, \Delta G_x}\right).
  \label{eq:flux-leak}
\end{equation}

The forward leakage rate $j_x^{+}$ depends on the membrane permeability
to $x$, the extracellular concentration of $x$, and the membrane
potential. We use a trapezoidal-barrier approximation to the membrane
energy profile~\cite{garlid1989}---a piecewise-linear model that captures
the voltage dependence of one-way ion flux without committing to a
specific channel geometry---and write
\begin{equation}
  j_x^{+} \;=\; \frac{S}{V}\, P_x\, [x]_e\, f_b(u),
  \label{eq:forward-rate}
\end{equation}
where $P_x$ is the permeability, $S/V$ the surface-to-volume ratio,
$[x]_e$ the extracellular concentration, and $f_b(u)$ the dimensionless
voltage factor, with
\begin{equation}
  u \;\equiv\; -\, z_x\, \eta\, \Delta\psi,
  \label{eq:u-def}
\end{equation}
where $z_x$ is the ion's charge number and $\eta = F/RT$. Note that
$u$ is ion-specific: cations and anions of the same valence see voltage
factors of opposite sign. The explicit form of $f_b$ is given in
Appendix~\ref{appendix:D}. With Eqs.~(\ref{eq:flux-leak}) and~(\ref{eq:forward-rate}) we can
compute the leakage dissipation rate $\mathcal{P}_{\mathrm{leak}}$
explicitly. 

\paragraph{Two paradigmatic architectures.}
To show that~\eqref{eq:bound-informal} is architecture-independent we instantiate
it on two coupling topologies (Fig.~\ref{fig:Two_Models}). Whichever model we use, there are three types of transport reactions: leakage reactions $x_e \rightleftharpoons x_i$ and coupled transport reactions of which there are two types: PMF coupled transport of the form $\text{H}^+_e + x_i \rightleftharpoons \text{H}^+_i + x_e$ and metabolism-coupled transport $E^\star + x_i \rightleftharpoons E + x_e$.

In Model~A, each ion is transported by a reaction powered directly by intracellular
metabolism:
\begin{align}
\label{eq:Z.1}
\forall x \in \{\text{H}^+,\text{Na}^+,\text{K}^+,\text{Cl}^-\}: \quad
E^\star_r + \nu_{x,r}\, x_i &\rightleftharpoons E_r + \nu_{x,r}\, x_e\\
\label{eq:Z.2}
2\, \text{H}_2\text{O} &\rightleftharpoons \text{H}_3\text{O}^+ + \text{OH}^-
\end{align}
where $E^\star_r \rightarrow E_r$ is the energy-providing half-reaction
(e.g.\ ATP hydrolysis or NADH oxidation) and $\nu_{x,r}$ is the
stoichiometry of ion $x$ in reaction $r$. In Model~B, only protons are
pumped metabolically and the motive forces of Na$^+$, K$^+$, Cl$^-$ are
maintained by proton antiporters:
\begin{align}
\label{eq:Z.3}
E^\star + \nu_{\text{H}^+,M}\, \text{H}^+_e &\rightleftharpoons
E + \nu_{\text{H}^+,M}\, \text{H}^+_i\\
\label{eq:Z.4}
\forall x \in \{\text{Na}^+,\text{K}^+,\text{Cl}^-\}: \quad
\nu_{\text{H}^+,r}\, \text{H}^+_e + \nu_{x,r}\, x_i &\rightleftharpoons
\nu_{\text{H}^+,r}\, \text{H}^+_i + \nu_{x,r}\, x_e\\
\label{eq:Z.5}
2\, \text{H}_2\text{O} &\rightleftharpoons \text{H}_3\text{O}^+ + \text{OH}^-
\end{align}
We assume a single transport reaction per ion in each model; allowing
multiple redundant transporters does not change the bound, since at
minimum dissipation only one of any pair of redundant reactions is
active in a given electrochemical state~\cite{Terradot2024}. The next subsection turns the informal inequality~(\ref{eq:bound-informal}) into a quantitative bound on the active power required to maintain the full set of IMFs whic holds for the two canonical architectures described above.

\paragraph{A bound on energy expenditures for ion transport}

We now turn the informal inequality~\eqref{eq:bound-informal} of the previous paragraph into a quantitative statement. Our central claim is that, at steady state, the minimal power density diverted from intracellular metabolism to sustain an electrophysiological state is
\begin{equation}
  \label{eq:A.14}
  Q_T \equiv -F \sum_{r \in MCT} \Delta G_{E,r} j_r \geq -F \sum_{x \backslash \text{OH}^-} \Delta G_x j_x
\end{equation}
where MCT holds for metabolic-coupled transport reactions (one for each ion for Model A, reactions~\eqref{eq:Z.1} and a single one for protons in Model B, reaction~\eqref{eq:Z.3}) and where $\Delta G_{E,r}$ is the amount of energy consumed by intracellular metabolism per turnover of a MCT reaction of the form $E^\star + x_i \rightleftharpoons E + x_e$. In addition, the Faraday constant $F$ converts the volt-valued $\Delta G_x$
back to energy per mole of ion, $j_x$ is the leakage flux
(Eq.~\eqref{eq:flux-leak}), and $Q_T$ is a power \emph{density} (energy
per unit time per unit volume). Eq.~\eqref{eq:A.14} identifies the
minimal energetic cost of IMF maintenance with the rate of free-energy
dissipation associated with passive leakage, and it holds independently
of organism, energy source, or transporter architecture. We further show in appendix~\ref{appendix:E} that Eq.~\eqref{eq:A.14} holds for the two canonical architectures Model A and B.

\begin{figure}[h!]
\centering
\includegraphics[scale=0.5]{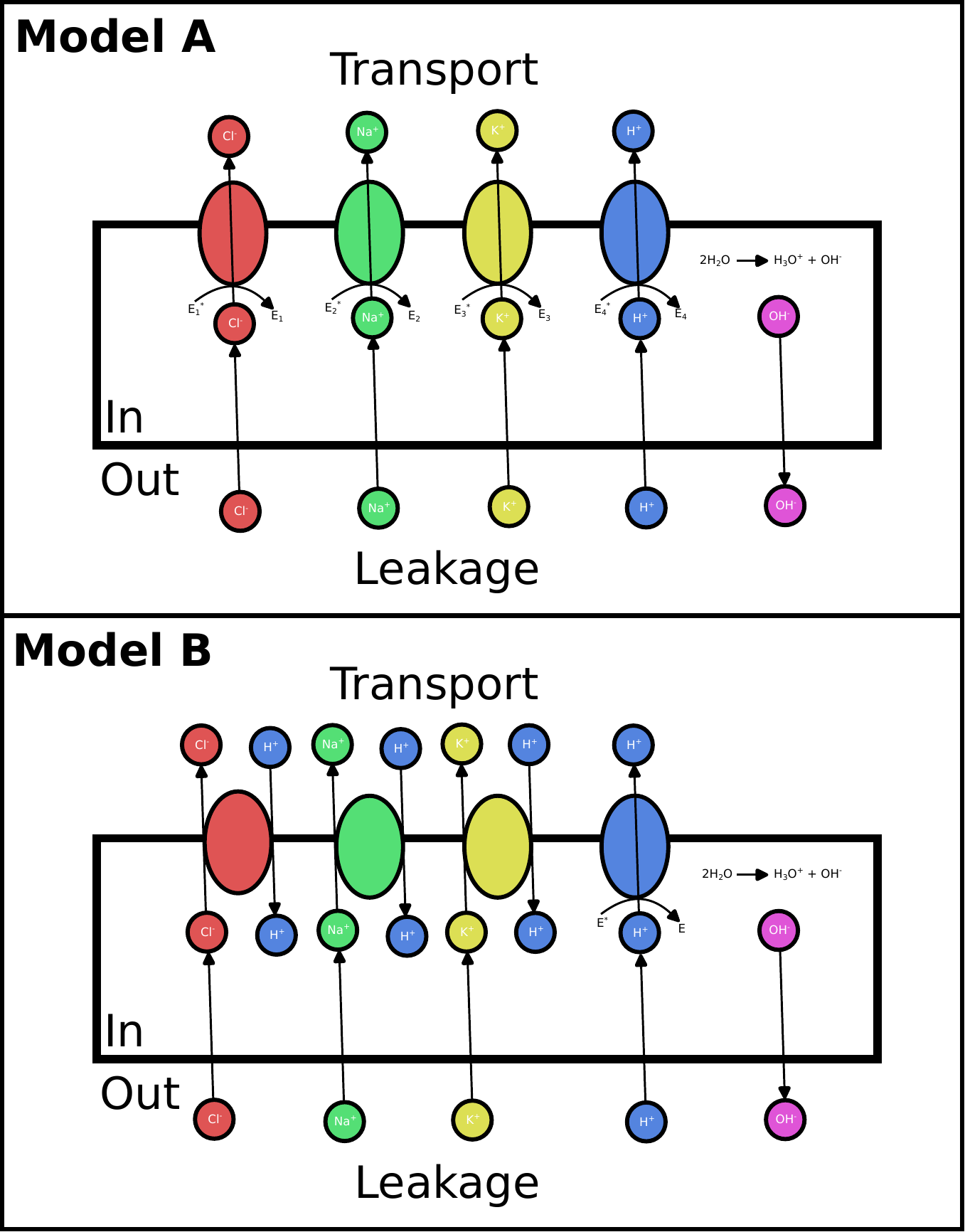}
\caption{\label{fig:Two_Models} Two canonical transport architectures
used as running examples throughout this work. Model~A powers every
ion transporter directly by an intracellular energy-providing
half-reaction (e.g.\ ATP hydrolysis or NADH oxidation). Model~B, the
\textit{E.\ coli} architecture, couples only proton extrusion to
intracellular metabolism; all other ions ride on proton:ion
antiporters. Despite their different coupling topologies, both models
saturate the same thermodynamic lower bound~\eqref{eq:A.14}.}
\end{figure}

\paragraph{Energetic efficiency.}
The energetic efficiency of a transport reaction is the fraction of
consumed free energy converted into electrochemical work. For reactions
coupled directly to intracellular metabolism (Eq.~\eqref{eq:Z.1} and
Eq.~\eqref{eq:Z.3}),
\begin{equation}
\label{eq:Z.6}
\epsilon_{x,r} = \dfrac{\nu_{x,r}\, \Delta G_x}{\Delta G_{E,r}},
\end{equation}
where $\Delta G_x$ is the IMF of ion $x$ and $\Delta G_{E,r}$ is the
free energy released by the energy-providing half-reaction. For
proton-coupled reactions in Model~B (Eq.~\eqref{eq:Z.4}),
\begin{equation}
\label{eq:Z.7}
\epsilon_{x,r} = \dfrac{\nu_{x,r}\, \Delta G_x}{\nu_{\text{H}^+,r}\, \Delta G_{\text{H}^+}}.
\end{equation}
Thermodynamic consistency requires $0 \le \epsilon_{x,r} \le 1$
(Appendix~\ref{appendix:E}): the upper bound is the idealized limit in which all
supplied free energy becomes electrochemical work.

\paragraph{Saturating the bound.}
In the idealized limit $\epsilon_{x,r} = 1$, Appendix~\ref{appendix:E} shows that
both Model~A and Model~B yield Eq.~\eqref{eq:A.14} for every ion
$x \neq \text{OH}^-$. (Dissipation due to hydroxide leakage is
negligible in all regimes we consider; see Appendix Table~\ref{table:Neglecting_OH}.) No
steady-state IMF maintenance can occur at a lower energetic cost,
whatever the transporter stoichiometry or the choice of power source.

\paragraph{Biological interpretation.}
Biologically, 100\% energetic efficiency is the evolutionary limit
reached by a cell adapted to a single, unchanging environment: such a
cell experiences a unique set of electrochemical potentials, and
selection tunes its transporter stoichiometries to match them. The
bound~\eqref{eq:A.14} is therefore the cost a perfectly adapted cell
would pay to fight passive leakage; real cells, which must cope with
variable environments and finite transporter repertoires, pay strictly
more (Section~\ref{sec:2.6}). With the bound in hand, we next relate IMFs to
membrane potential.

\subsection{A lower bound on membrane potential maintenance cost}
\label{sec:2.2}

Unlike Ref.~\cite{Terradot2024}, which treated a single cation, we now
consider three permeant species---two cations (Na$^+$, K$^+$) and one
representative anion (Cl$^-$)---plus intracellular captive charges
that contribute to $\Delta\psi$ but do not cross the membrane. The
steady-state voltage equation (Appendix~\ref{appendix:C})
generalizes to
\begin{equation}
\label{eq:A.6}
0 \;=\; \alpha_{\text{Y}} +
\sum_{x \in \{\text{Na}^+,\text{K}^+,\text{Cl}^-\}}
z_x\, \alpha_x\, e^{\eta\, (\Delta G_x - z_x\, \Delta\psi)},
\end{equation}
where
\begin{align*}
\alpha_{\text{Y}} &\equiv \frac{z_{\text{Y}}\,[\text{Y}]_i}{[\text{Ion}]_e}
&&\text{(captive charge density, normalized)},\\
[\text{Ion}]_e &\approx [\text{Na}^+]_e + [\text{K}^+]_e + [\text{Cl}^-]_e
&&\text{(total extracellular ion concentration)},\\
\alpha_x &\equiv \frac{[x]_e}{[\text{Ion}]_e}
&&\text{(fractional abundance of ion $x$)}.
\end{align*}
For a fixed extracellular ratio $\beta \equiv \alpha_{\text{Na}^+} /
\alpha_{\text{K}^+}$, many (SMF, KMF) pairs sustain the same
$\Delta\psi$. Appendix~3 enumerates the compatible pairs; here we
minimize over them.

\paragraph{Minimal cost of maintaining membrane potential.}

Because $\Delta\psi$ is the aggregate outcome of sustaining SMF, KMF,
and the chloride motive force simultaneously and is independent of the PMF, the minimal energetic
cost of a prescribed $\Delta\psi$ therefore is
\begin{equation}
\label{eq:A.15}
\min\!\left(Q_{\Delta\psi}\right) =
\min\!\left(Q_{\text{Na}^+}\right) +
\min\!\left(Q_{\text{K}^+}\right) +
\min\!\left(Q_{\text{Cl}^-}\right),
\end{equation}
where we ignored the term for PMF maintenance $Q_{\text{H}^+}$ in Eqs.~\eqref{eq:W.11c} and \eqref{eq:W.25} for both model A and B and assumed maximal energetic efficiencies. Upon these assumptions, each term of Eq.~\eqref{eq:A.16} is given by
\begin{equation}
\label{eq:A.15b}
\forall x \in \{\text{Na}^+,\text{K}^+,\text{Cl}^-\}: \min(Q_x) = - F \Delta G_x j_x
\end{equation}
as shown in lemma~\ref{lem:per-ion}, Eq.~\eqref{eq:W.3}. Taking $P_{\text{K}^+} >
P_{\text{Na}^+}$ (Table~\ref{table:constants}), we compute
$\min\!\left(Q_{\Delta\psi}\right)$ and express it as a fraction of
\begin{equation}
\label{eq:A.16}
Q_0 \;\equiv\; 3.14 \times 10^{-13}\,\text{W per cell},
\end{equation}
the reference respiratory power of a single \textit{E.\ coli} cell at
moderate growth rate (Appendix~\ref{appendix:G}). Figure~\ref{fig:1} displays the
result. Two robust trends emerge, and two predictions follow.

\paragraph{Trend 1: low $\beta$ raises dissipation.}
Decreasing $\beta$ (more extracellular K$^+$ than Na$^+$) increases
the cost. Because K$^+$ is the more permeable cation, its passive
leakage flux grows faster with the extracellular pool, and
compensating it demands more active transport.

\paragraph{Trend 2: captive charges widen the cost-free regime.}
Increasing $|\alpha_{\text{Y}}|$ enlarges the range of $\Delta\psi$
that can be sustained with all IMFs set to zero---the Donnan regime
in which no active transport is needed and the maintenance cost
collapses to zero. Outside this regime the cost grows with
$|\Delta\psi|$, as expected.

\paragraph{Prediction: selective pressure against Eyring-like
leakage.}
Comparing Goldman--Hodgkin--Katz (GHK) and Eyring leakage under
identical permeabilities, the Eyring model predicts substantially
larger dissipation: maintaining $\Delta\psi < -140$~mV would exceed
the entire respiratory budget $Q_0$. This suggests an evolutionary
pressure for GHK-like leakage kinetics whenever cells must sustain
strong polarization.

\paragraph{Prediction: salinity stress lowers $\Delta\psi$.}
Equation~\eqref{eq:forward-rate} shows that raising $[\text{Ion}]_e$, therefore $[x]_e$ for sodium or potassium,
increases the forward leakage flux, and hence the cost, of sustaining
any fixed $\Delta\psi$. Under salt stress, cells that cannot afford
the extra cost should depolarize. This matches direct observations in
\textit{E.\ coli}, where elevated NaCl concentrations induce
membrane depolarization~\cite{Schabala2009}.

\begin{figure}[h!]
\centering
\includegraphics[scale=0.4]{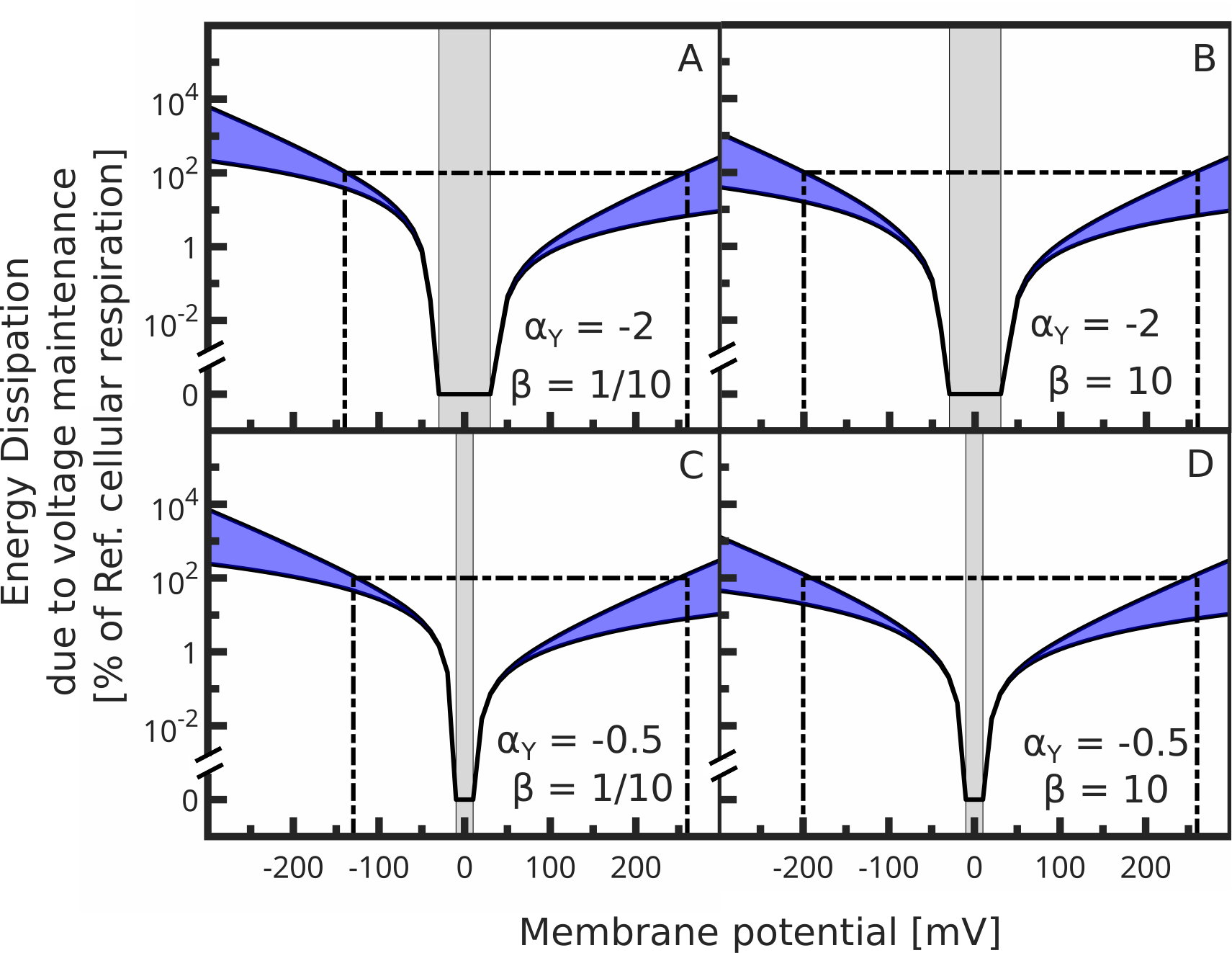}
\caption{\label{fig:1}
Minimal dissipation $Q_{\Delta\psi}$ required to sustain a membrane
potential $\Delta\psi$, as a fraction of the reference respiratory
power $Q_0$ of a single \textit{E.\ coli} cell
(Eqs.~\eqref{eq:A.15},\eqref{eq:A.16}), for several extracellular
cation ratios $\beta = \alpha_{\text{Na}^+}/\alpha_{\text{K}^+}$ and
captive-charge contributions $\alpha_{\text{Y}}$. Lower (GHK) and
upper (Eyring) bounds correspond to the two leakage laws. Grey
shading marks the Donnan regime, in which $\Delta\psi$ is sustained
with zero ionic motive forces. Higher $\beta$ reduces cost because
sodium is less permeable than potassium. Increased extracellular
salinity raises dissipation via enhanced cation leakage.}
\end{figure}

The trends above raise the next question: given a target
$\Delta\psi$, which partitioning of the total cost across SMF and
KMF is optimal, and what intracellular composition does it predict?
We address this in Section~2.3.

\subsection{Optimal CMFs and the resulting intracellular composition}
\label{sec:2.3}

Minimizing the maintenance cost~\eqref{eq:A.15} of a prescribed
$\Delta\psi$ yields a unique cellular state, characterized by
systematic asymmetry between the sodium and potassium motive forces
and by a K$^+$-rich, Na$^+$-poor cytoplasm that matches measurements
in \textit{Saccharomyces cerevisiae}~\cite{Sunder1996,Olz1993} and
\textit{Escherichia coli}~\cite{MEURY1981,Epstein1966}. Importantly,
this pattern emerges from permeability asymmetry alone: low
intracellular sodium need not be attributed to sodium toxicity, but
falls out directly as the cost-minimizing solution.

\paragraph{Optimal cationic motive forces.}
For a fixed $\Delta\psi$ and captive-charge contribution
$\alpha_{\text{Y}}$, we first compute the single-cation
cationic motive force (CMF) required to sustain $\Delta\psi$
(Appendix~Eq.~\ref{eq:C8}). At a given extracellular ratio
$\beta$, multiple (SMF, KMF) pairs realize the same CMF (Appendix~\ref{appendix:C}, figure~\ref{Appendix_fig:figure5});
among these we select the pair that minimizes $W_{\Delta\psi}$
(Appendix~Fig.~\ref{Appendix_fig:figure6}).

The optimum (Fig.~\ref{fig:2}A--D) shows that at any prescribed
$\Delta\psi$ the optimal SMF is more negative than the optimal KMF.
The reason is direct: K$^+$ is the more permeable cation, so K$^+$
leakage contributes more to dissipation per unit of driving force,
and the cost-minimizer responds by allocating a smaller KMF and a
larger SMF. Intracellular ionic composition is therefore the
solution of a constrained dissipation-minimization problem, not an
arbitrary physiological set point.

Two further parameter dependences matter:
\begin{itemize}
  \item Increasing $|\alpha_{\text{Y}}|$ shifts the required CMFs
        toward zero at fixed $\Delta\psi$, reducing the magnitude of
        the SMF and KMF that must be actively maintained.
  \item Raising $\beta$ (more extracellular Na$^+$ than K$^+$)
        pushes the optimal SMF and KMF closer to zero, because
        large cation gradients become energetically less favorable.
\end{itemize}
Sufficiently strong polarization requires active export of both
Na$^+$ and K$^+$ (Appendix~Fig.~\ref{Appendix_fig:figure7}).

\paragraph{Intracellular ion concentrations.}
Intracellular concentrations follow directly from the
definition of the electrochemical potential:
\begin{equation}
\label{eq:A.31}
\Delta G_x \;=\; z_x\, \Delta\psi + \eta^{-1}\ln\!\left(\frac{[x]_i}{[x]_e}\right),
\end{equation}
which rearranges (using the volts convention for $\Delta G_x$) to
\begin{equation}
\label{eq:A.32}
[x]_i \;=\; [x]_e\, e^{\eta\,(\Delta G_x - z_x\, \Delta\psi)}.
\end{equation}
Feeding the cost-minimizing (SMF, KMF) pair into~\eqref{eq:A.32}
yields the compositions in Fig.~\ref{fig:2}E--J.

\paragraph{Empirical matches.}
Three of the four patterns shown are direct comparisons with
published data:
\begin{enumerate}
  \item[(i)] At negative $\Delta\psi$, intracellular chloride is
        low, matching~\cite{Schultz1962}; it rises as the cell
        depolarizes.
  \item[(ii)] Intracellular K$^+$ exceeds Na$^+$ across broad
        parameter ranges, matching
        \textit{S.~cerevisiae}~\cite{Sunder1996,Olz1993} and
        \textit{E.~coli}~\cite{MEURY1981,Epstein1966}. The
        K$^+$/Na$^+$ dominance is predicted to weaken as cells
        depolarize into stationary phase (Appendix~Fig.~9),
        matching~\cite{Schultz1961}.
\end{enumerate}
Two further patterns are predictions that follow from the
optimization:
\begin{enumerate}
  \item[(iii)] Raising $\beta$ (extracellular Na:K ratio) raises the
        intracellular Na:K ratio, but K$^+$ remains dominant as long
        as the permeability asymmetry is preserved
        (Appendix~Fig.~\ref{Appendix_fig:figure9} and \ref{Appendix_fig:figure8}).
  \item[(iv)] Raising $|\alpha_{\text{Y}}|$ increases intracellular
        concentrations of all three ions at the optimum, shifting
        the optimum toward higher intracellular ionic strength.
\end{enumerate}

These results rely only on the experimentally supported assumption
$P_{\text{K}^+} > P_{\text{Na}^+}$. If this asymmetry is general
across living cells, the framework predicts the same K$^+$-rich,
Na$^+$-poor cytoplasm across taxa, \emph{as a direct thermodynamic
consequence of minimizing membrane-potential maintenance cost.}

\begin{figure}[h!]
\centering
\includegraphics[scale=0.20]{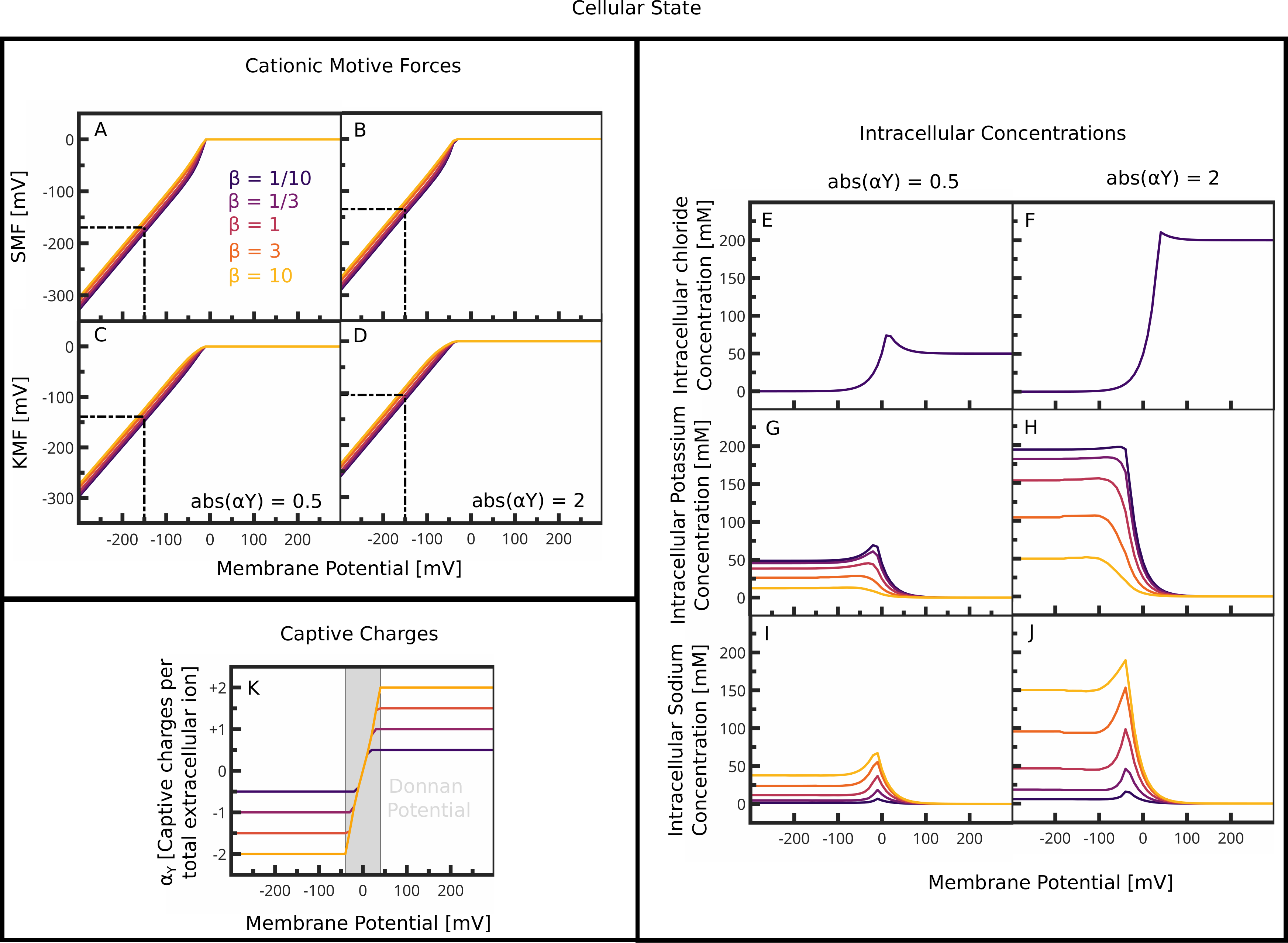}
\caption{\label{fig:2}
Cost-minimizing cellular states for a prescribed membrane potential
$\Delta\psi$. (A--D) Optimal SMF and KMF for varying $\beta$ and
$|\alpha_{\text{Y}}|$: SMF is systematically more negative than KMF
because $P_{\text{K}^+} > P_{\text{Na}^+}$; increasing
$|\alpha_{\text{Y}}|$ shrinks the gradients required. (E--J)
Intracellular Na$^+$, K$^+$, Cl$^-$ concentrations at the optimum;
K$^+$ dominates across broad parameter ranges, and Cl$^-$ rises with
depolarization. (K) The cost-minimizing $\alpha_{\text{Y}}$ as a
function of $\Delta\psi$; within the Donnan band (grey,
$|\alpha_{\text{Y}}| = 2$) optimal maintenance occurs at zero
ionic motive forces. Bottom line: permeability asymmetry alone is
sufficient to explain the universal K$^+$-rich, Na$^+$-poor
cytoplasm.}
\end{figure}

\subsection{Extremophiles share a high energetic cost of maintaining
internal pH and membrane potential}
\label{sec:2.4}

Acidophiles, alkalinophiles, halophiles, and thermophiles share a
single thermodynamic feature: they all operate at systematically
elevated unavoidable dissipation relative to moderate environments.
The underlying ionic mechanisms differ, but the thermodynamic bottom
line is the same.

The total maintenance cost is the sum of the cost of sustaining the
membrane potential and that of sustaining the proton motive force,
\begin{equation}
\label{eq:A.35}
\min\!\left(Q_T\right) \;=\;
\min\!\left(Q_{\text{H}^+}\right) +
\min\!\left(Q_{\Delta\psi}\right),
\end{equation}
where $Q_{\Delta\psi}$ is given by Eq.~\eqref{eq:A.15} and
$Q_{\text{H}^+}$ by Eq.~\eqref{eq:W.11}, see appendix~\ref{appendix:E}. We compute
$\min\!\left(Q_T\right)$ under prescribed intracellular pH, PMF, and
osmotic pressure (Appendix~\ref{appendix:H} for osmotic implementation; Appendix~\ref{appendix:C}
for the $Q_{\Delta\psi}$ optimization).

Figure~\ref{fig:3} and Table~\ref{table:1} summarize how each
environmental axis reshapes the dissipation budget:
\begin{itemize}
  \item \textbf{Acidic extracellular pH.} Dissipation is dominated
        by $Q_{\text{H}^+}$: a large inward proton leakage must be
        balanced at steady state. Cation leakage is negligible,
        and the cost is insensitive to salinity and osmotic
        pressure.
  \item \textbf{Alkaline extracellular pH.} Dissipation is
        dominated by cation leakage: sustaining internal pH and
        PMF now demands strongly negative $\Delta\psi$, which
        amplifies Na$^+$ and K$^+$ leakage. Cost rises sharply
        relative to neutral conditions.
  \item \textbf{High salinity.} Forward leakage flux scales with
        $[\text{Ion}]_e$ (Eq.~\eqref{eq:forward-rate}), so the cost of
        sustaining $\Delta\psi$ grows roughly linearly with
        extracellular ionic strength---provided cation leakage is
        rate-limiting. At acidic pH, where proton leakage
        dominates, salinity has little effect.
  \item \textbf{High temperature.} Elevated temperature raises the
        ionic permeability of the
        membrane~\cite{Toyoshima1975,Driessen1996}, and minimal
        maintenance power is proportional to leakage flux, so
        thermophilic conditions pay a permeability-scaled cost
        increase across all ions.
\end{itemize}
In all four regimes, the energetic burden arises from the same
thermodynamic principle: at steady state, active transport must
compensate for free-energy dissipation by passive leakage.
Extremophily entails sustained operation closer to the
thermodynamic ceiling of maintenance power.

\paragraph{Osmotic pressure as a partial lever.}
Table~\ref{table:1} quantifies the effect of imposed osmotic
pressure. At acidic pH, $Q_T$ is insensitive to $\Delta\Pi$ (proton
leakage dominates and does not depend on osmotic pressure). At
neutral and alkaline pH, where cation leakage dominates, raising
$\Delta\Pi$ reduces the $\Delta\psi$ required for homeostasis and
hence lowers $Q_{\Delta\psi}$. Osmotic adjustment is therefore a
partial mitigation: it can modulate maintenance cost but cannot
eliminate it.

\begin{table}[h!]
\begin{center}
\begin{tabular}{|c|c|c|c|}
\hline
Environment & $\Delta \Pi = 0$ & $\Delta \Pi = 10$ atm & $\Delta \Pi = 20$ atm\\
\hline
$\text{pH}_e = 3$, $[\text{Ion}]_e = 100$ mM & $2.5 \times 10^3$  & $2.5 \times 10^3$  & $2.5 \times 10^3$  \\
\hline
$\text{pH}_e = 7$, $[\text{Ion}]_e = 100$ mM & $30$ & $25$  & $23$ \\
\hline
$\text{pH}_e = 11$, $[\text{Ion}]_e = 100$ mM & $129$ & $116$  & $112$  \\
\hline
$\text{pH}_e = 3$, $[\text{Ion}]_e = 1000$ mM & $2.5 \times 10^3$ & $2.5 \times 10^3$ & $2.5 \times 10^3$ \\
\hline
$\text{pH}_e = 7$, $[\text{Ion}]_e = 1000$ mM & $262$ & $205$  & $184$  \\
\hline
$\text{pH}_e = 11$, $[\text{Ion}]_e = 1000$ mM & $1.3 \times 10^3$  & $1.2 \times 10^3$ & $1.1 \times 10^3$\\
\hline
\end{tabular}
\caption{\label{table:1}$Q_T$ (GHK leakage) as a function of
osmotic pressure, across three extracellular pH values and two
salinities. At acidic pH$_e$, proton leakage dominates and the
cost is osmotic-insensitive. At neutral and alkaline pH$_e$,
cation leakage dominates and higher osmotic pressure lowers the
IMFs required to maintain $\Delta\psi$, reducing the cost.}
\end{center}
\end{table}

\begin{figure}[h!]
\centering
\includegraphics[scale=0.3]{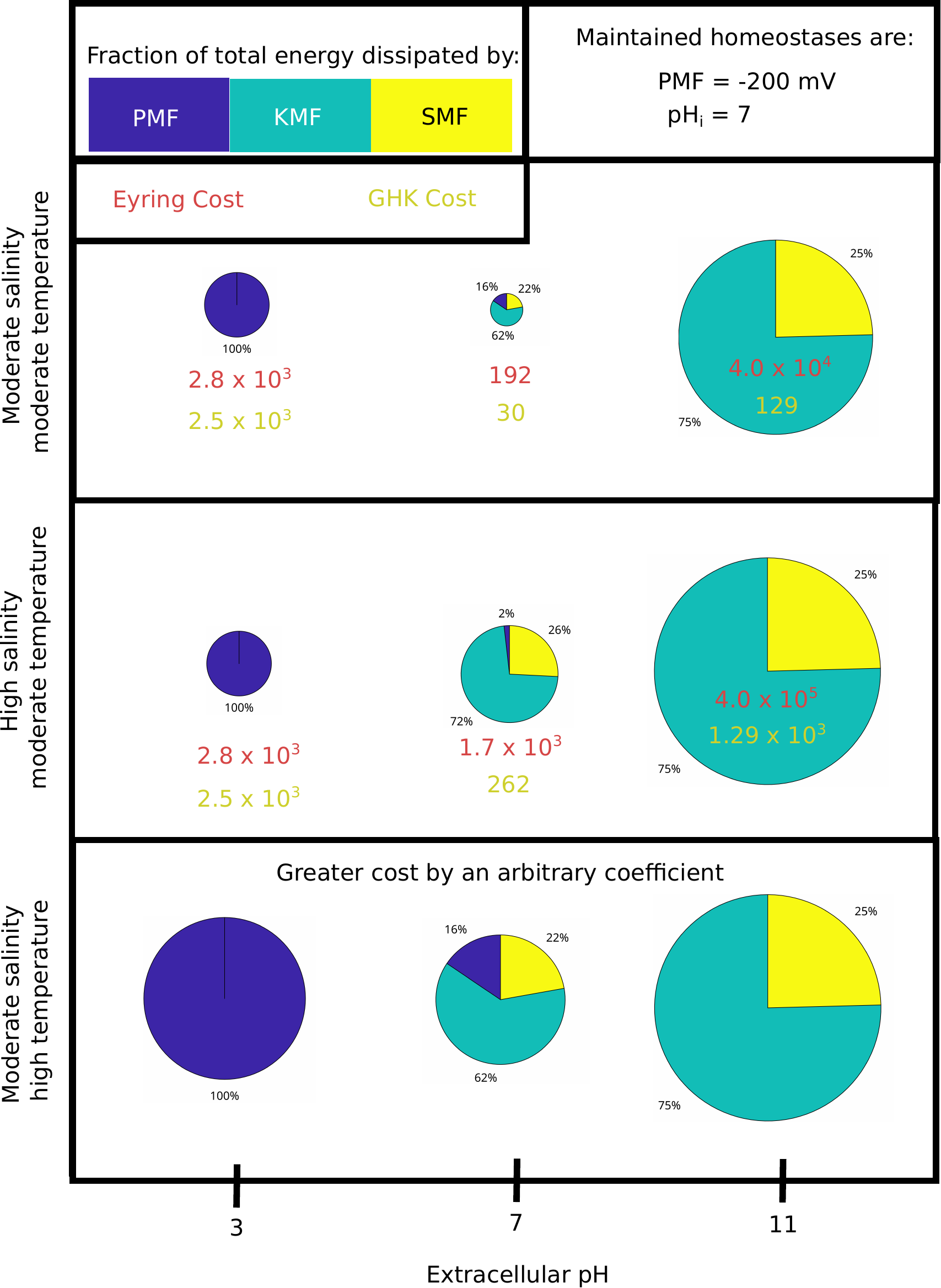}
\caption{\label{fig:3}
Extremophilic regimes require systematically elevated maintenance
power, via four distinct ionic mechanisms. Total minimal maintenance
dissipation $\min(Q_T)$ needed to sustain pH$_i = 7$ and PMF $=-200$~mV at
zero osmotic pressure, across extracellular pH, salinity, and
temperature. Average captive-charge valency in $[-2,2]$, $\beta=1$.
Moderate salinity: $[\text{Ion}]_e = 100$~mM; high salinity:
$1000$~mM. Temperature increase is modelled as a permeability
increase. Pie charts indicate the relative contributions of PMF,
KMF, and SMF to total dissipation. Red and yellow values correspond
to Eyring and GHK leakage, respectively. Acidic regimes are
dominated by proton dissipation; alkaline and high-salinity regimes
by cation dissipation.}
\end{figure}

\subsection{Maintenance cost differentiates metabolic regimes}
\label{sec:2.5}

Respiration requires a large PMF; fermentation does not. When
environmental stress makes sustaining a large PMF prohibitively
dissipative, lowering the PMF saves maintenance power---and
metabolic regimes that tolerate a reduced PMF become energetically
favored. This provides a thermodynamic rationale for the shift from
respiration to fermentation under stress.

\paragraph{The energetic asymmetry.}
Respiratory architectures synthesize ATP through ion-translocating
electron-transport chains and therefore require sustained membrane
polarization and a large PMF. Fermentative architectures rely on
substrate-level phosphorylation and can operate with smaller
transmembrane gradients. Because the maintenance cost derived in
Sections~\ref{sec:2.2}--\ref{sec:2.4} scales with the magnitude of the gradients, a
metabolic regime that demands a large PMF carries, as a byproduct, a
larger unavoidable dissipation floor. Environmental conditions that
amplify this floor---extreme pH, high salinity, elevated membrane
permeability---therefore penalize respiration more than
fermentation.

\paragraph{Direct empirical match.}
Cells reduce the magnitude of their PMF under highly alkaline
extracellular conditions, precisely the regime in which our
framework predicts that sustaining a large PMF becomes energetically
prohibitive~\cite{Terradot2024}. By lowering the PMF, the cell cuts
its maintenance burden, at the cost of a smaller driving force for
proton-coupled processes.

\paragraph{From PMF reduction to metabolic switching.}
Our framework directly predicts only that \emph{cells should operate
at a lower PMF under stress}. The further step---switching from
respiration to fermentation---is a physiological consequence: if
ATP production via respiration depends on a PMF the cell can no
longer afford, substrate-level phosphorylation becomes the
alternative. In this view, membrane physics constrains metabolic
organization: environmental conditions set the maintenance cost of
the electrophysiological configuration a given regime requires, and
that cost feeds back into the viable space of metabolic strategies
(Figure~\ref{fig:pmf_regimes}).

\paragraph{Other dimensions of the trade-off.}
Maintenance cost is only one of four factors selecting between
fermentation and respiration. The others are: (i) different
by-products and environmental feedbacks; (ii) different proteome
cost per unit ATP produced~\cite{Basan2015bis}, which favors
fermentation; (iii) different yield per substrate, which favors
respiration. Our framework adds a fourth consideration, pulling in
the same direction as (ii) under stress.

\begin{figure}[h!]
\centering
\includegraphics[scale=0.30]{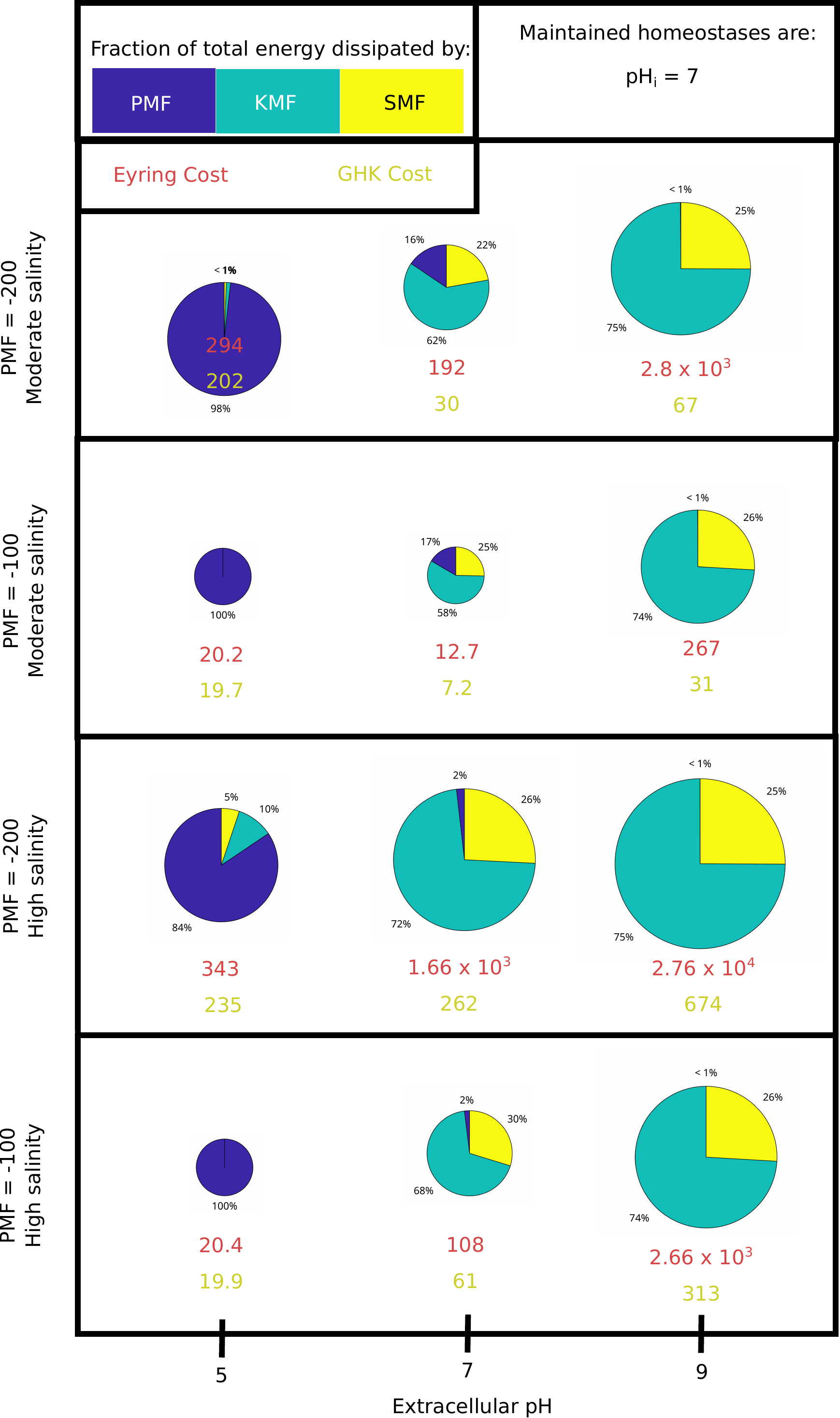}
\caption{\label{fig:pmf_regimes}
Total maintenance cost $Q_T$ needed to sustain pH$_i=7$ at PMF
$=-200$~mV (respiratory) versus PMF $=-100$~mV (reduced-gradient),
across extracellular pH and salinity. Captive charge valency in
$[-2,2]$, $\beta = 1$. Lowering the PMF reduces the $\Delta\psi$ that
must be maintained, cutting cation leakage and total dissipation---
but narrows the range of extracellular pH compatible with
intracellular neutrality~\cite{Terradot2024}. Eyring (red) and GHK
(yellow) leakage models shown.}
\end{figure}

We have assumed throughout Sections~\ref{sec:2.1}--\ref{sec:2.5} that transporters
operate at $100\%$ energetic efficiency. Section~\ref{sec:2.6} shows that this
limit is generally unattainable in practice, and that the resulting
sub-optimality is itself a key shaper of transporter diversity.

\subsection{Why perfect energetic efficiency is not attainable}
\label{sec:2.6}

Three distinct, independent constraints prevent real cells from
reaching the thermodynamic lower bound derived in Section~\ref{sec:2.1}
(i)~a trade-off between energetic and catalytic efficiency, (ii)~the
finite diversity of transporters, and (iii)~the discrete
stoichiometries and energy quanta available to ion transport. The
first constraint applies to every single transporter individually;
the second and third apply to the cell as a whole and are amplified
in heterogeneous environments.

\paragraph{(i) Energetic vs catalytic efficiency.}
Define the catalytic efficiency $c_r$ of reaction $r$ as the
fraction by which the forward flux exceeds the reverse,
\begin{equation}
\label{eq:A.36}
c_r \;=\; 1 - e^{\eta\, \Delta G_r} \;=\; 1 - \frac{j_r^{-}}{j_r^{+}}
\end{equation}
so that Eq.~\eqref{eq:flux-generic} rewrites as $j_r = j_r^{+} c_r$. High
$c_r$ requires operation far from equilibrium ($\Delta G_r$ strongly
negative); high energetic efficiency $\epsilon_r$ requires operation
near equilibrium ($\Delta G_r \to 0$). The two requirements conflict.
Solving for $\epsilon_r$ in terms of the driving force
$\Delta G_E$ and $c_r$ yields
\begin{equation}
\label{eq:A.40}
\epsilon_r \;=\; 1 - \frac{1}{\eta\, \Delta G_{E,r}}\, \ln(1 - c_r)
\end{equation}
which quantifies the trade-off (Figure~\ref{fig:4}A--B). Because
maintaining non-zero steady-state transport flux demands operation
away from equilibrium, $\epsilon_r < 1$ necessarily.

\paragraph{(ii) Finite transporter diversity.}
A transporter $r$ has a discrete driving force $\Delta G_{E,r}$ and a fixed
stoichiometry. With a single transporter, energetic efficiency
reaches $100\%$ only when the IMF to be maintained exactly matches
$\Delta G_{E,r}$ (Figure~\ref{fig:4}C). Expanding the transporter
repertoire expands the set of IMFs maintained at high efficiency
(Fig.~\ref{fig:4}D--E), but achieving near-$100\%$ efficiency across
a \emph{continuum} of IMFs would require an infinite repertoire. The
broader the range of IMFs a cell must maintain across environments,
the more transporters it needs---and the further it falls from the
bound whenever its actual repertoire falls short.

\paragraph{(iii) Discrete stoichiometries and energy quanta.}
Transport is powered by discrete energy quanta (ATP hydrolysis
$\approx -560$~mV equivalent; PMF $\approx -200$~mV) and uses
integer stoichiometries. Only a discrete set of effective driving
forces can therefore be realized, and only a corresponding discrete
set of IMFs can be maintained at exact $100\%$ efficiency
(Fig.~\ref{fig:4}F--G). Maximizing efficiency across environments
therefore requires diversity in both stoichiometry and
energy-coupling mechanism.

\paragraph{Consequence for the next section.}
Together, (i)--(iii) elevate real maintenance cost above the
Section~\ref{sec:2.1} bound; in heterogeneous environments they select for
transporter diversity. Section~\ref{sec:2.7} shows how this selection
partitions transporter architectures across taxa.

\paragraph{Model A vs.\ Model B: the flexibility-efficiency trade.}
To see the partitioning, compare the total dissipation for the two
architectures of Fig.~\ref{fig:Two_Models}. Appendix~\ref{appendix:E} shows
\begin{align}
\label{eq:Y.1}
\dfrac{Q_T^{A}}{F} &= \sum_{x \backslash \text{OH}^-}
\frac{1}{\epsilon_x}\, \Delta G_x\, j_x\\
\label{eq:Y.2}
\dfrac{Q_T^{B}}{F} &= \frac{1}{\epsilon_M}
\left(\Delta G_{\text{H}^+}\, j_{\text{H}^+} +
\sum_{x \backslash \{\text{H}^+,\text{OH}^-\}}
\frac{1}{\epsilon_x}\, \Delta G_x\, j_x\right)
\end{align}
where $\epsilon_M$ is the efficiency of the metabolically coupled proton-pumping
reaction and $\epsilon_x$ the efficiency of the sole active transport reaction moving $x$ across the membrane. In Model~B, every non-proton transport pays two
efficiency losses in series---the metabolic pump $M$ and the
proton:ion antiporter---hence the two efficiencies multiply, making
Model~B intrinsically more dissipative. However, Model~B's
effective driving force is smaller (the PMF quantum rather than the
ATP quantum, Fig.~\ref{fig:4}F--G), which permits near-$100\%$
efficiency at weaker IMFs. The result is a clean trade-off:
Model~A wins in stable environments that demand large IMFs;
Model~B wins in variable environments that sometimes demand small
IMFs, at the cost of intrinsic dissipative overhead.

\begin{figure}[h!]
\centering
\includegraphics[scale=0.25]{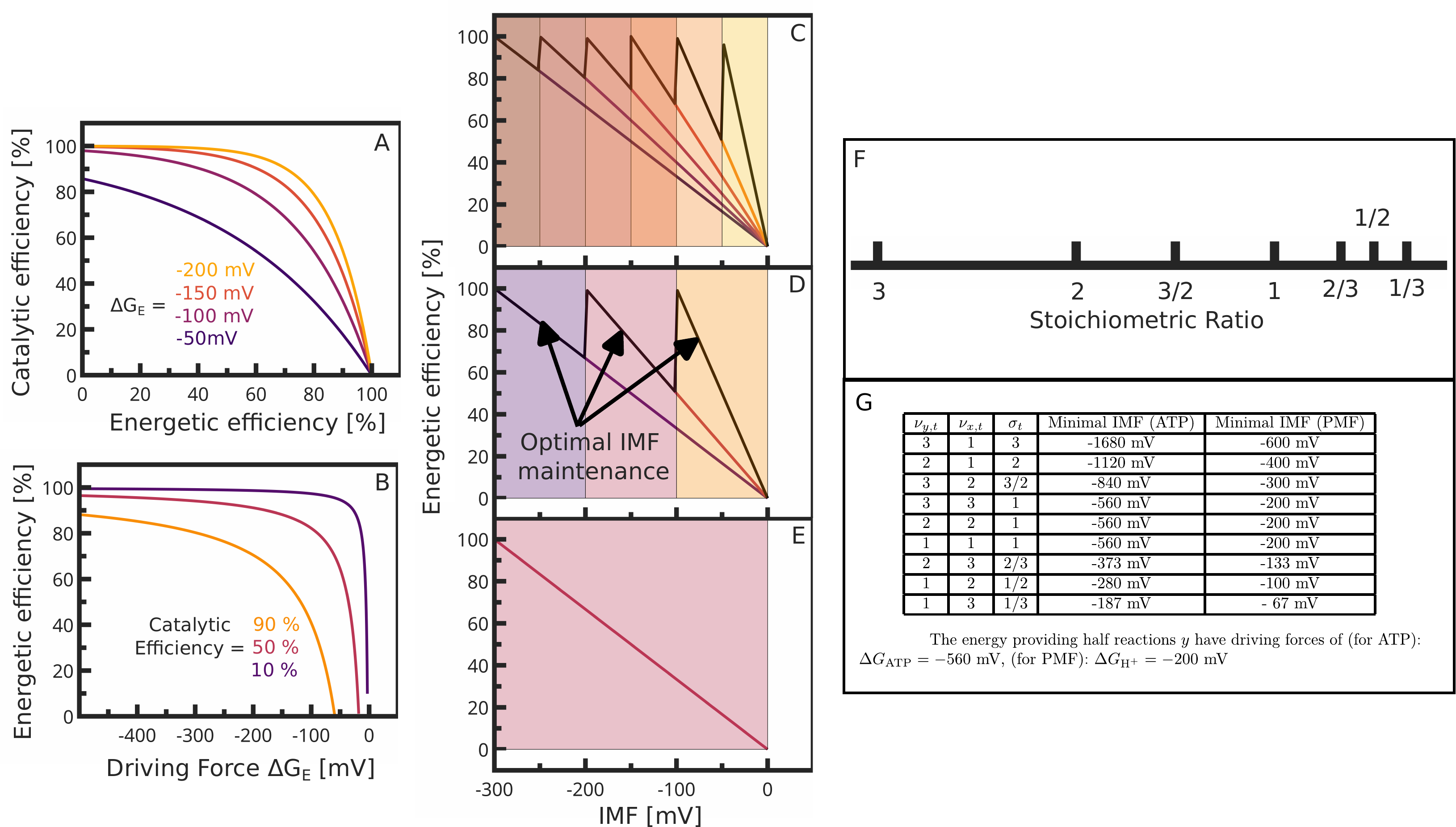}
\caption{\label{fig:4}
Real cells operate above the Section~2.1 bound for three
independent reasons. (A--B) Catalytic/energetic efficiency
trade-off: maintaining finite steady-state flux requires operation
away from equilibrium and hence $\epsilon_r < 1$. (C--E) Finite
transporter repertoires limit the range of IMFs maintained near
$100\%$ efficiency; broader environmental variability requires more
transport systems. (F--G) Discrete stoichiometries and quantized
energy sources (ATP $\approx -560$~mV; PMF $\approx -200$~mV)
restrict the achievable driving forces, so only specific IMF values
can be maintained at exact $100\%$ efficiency.}
\end{figure}

\subsection{Environmental constraints select distinct transport
architectures}
\label{sec:2.7}

Environmental pH selects qualitatively distinct transport
architectures. The range of pH a cell encounters determines the
diversity of its transporter repertoire, and the energetic cost of
the required $\Delta\psi$ dictates which energy currency
(PMF vs.\ ATP) its transporters exploit.

\paragraph{Neutrophiles.}
Neutrophiles such as \textit{E.\ coli} must maintain cytoplasmic
neutrality across a broad range of extracellular pH and therefore
carry a correspondingly broad transporter repertoire
(Fig.~\ref{fig:5}A). At strongly acidic pH they sustain a positive
$\Delta\psi$ via anion motive force (AMF) mechanisms; near
neutrality, moderate cation driving forces suffice; at alkaline pH
they need cation transporters with high driving force. Per
Eq.~\eqref{eq:Y.2}, a lower energy quantum (PMF rather than ATP)
permits the cell to operate near $100\%$ efficiency across the
widest range of IMFs, at the cost of higher intrinsic dissipation---
a choice that fits the variable-environment profile of a
neutrophile.

As shown in Section~2.3, dissipation minimization favors a weaker
KMF than SMF because $P_{\text{K}^+} > P_{\text{Na}^+}$.
Mechanistically, this predicts lower H$^+$:K$^+$ stoichiometry in
K$^+$ antiporters than H$^+$:Na$^+$ stoichiometry in Na$^+$
antiporters---a stoichiometric signature of the optimization that
should be experimentally observable.

\paragraph{Acidophiles.}
Acidophiles inhabit persistently acidic environments, where the
positive $\Delta\psi$ needed to sustain internal neutrality and
PMF is comparatively modest (Appendix~Fig.~9). A single
low-driving-force AMF system can suffice (Fig.~\ref{fig:5}B).
The trade-off is environmental range: this minimal architecture
cannot cope with neutral or alkaline conditions, so specialization
buys energetic economy at the cost of robustness.

\paragraph{Alkalinophiles (1/3): ATP-driven rather than
PMF-driven.}
Alkalinophiles face the most extreme electrophysiological
constraints. Sustaining internal pH and PMF in highly alkaline
environments requires strongly negative $\Delta\psi$
(Appendix~Fig.~\ref{Appendix_fig:figure10}), and the corresponding cation motive forces
exceed what PMF-coupled antiporters can produce in
Model~B~\eqref{eq:Y.2}. Alkalinophiles therefore frequently couple
IMF maintenance to ATP hydrolysis (Fig.~\ref{fig:5}C): the larger
energy quantum of ATP enables the required driving forces. The
cost however incurs additional efficiency losses if ATP is produced through respiration---proton extrusion, ATP
synthesis, and ATP-driven transport each cost some energy---
whereas PMF-driven transport bypasses one conversion step.

\paragraph{Alkalinophiles (2/3): lowering PMF as a partial
mitigation.}
At alkaline pH, operating at a reduced PMF lowers the $\Delta\psi$
required for homeostasis and hence the cation leakage cost. The
trade-off: a lower PMF weakens the driving force available to
proton-coupled efflux systems, which in turn limits their ability
to sustain the required SMF and KMF. Cells at alkaline pH therefore
face a structural trade-off---lower PMF saves maintenance power
but forces a shift toward ATP-coupled transport architectures.

\paragraph{Alkalinophiles (3/3): stoichiometry and the dissipation
minimum.}
The dissipation argument extends to ATP-driven systems. At steady
state, the maintenance power for KMF equals the product of the
leakage flux and the free energy invested per K$^+$ exported.
Lower K$^+$:ATP stoichiometry reduces that per-ion free-energy
investment and therefore directly lowers the steady-state power
dissipated in maintaining the K$^+$ gradient. Because the
Section~\ref{sec:2.3} optimum occurs at weaker KMFs, ATP-driven K$^+$ pumps
operating at low K$^+$:ATP stoichiometry shift membrane-potential
maintenance toward the dissipation minimum. This is a general
consequence of the Section~\ref{sec:2.3} result that holds regardless of the
energy currency.

\begin{figure}[h!]
\centering
\includegraphics[scale=0.5]{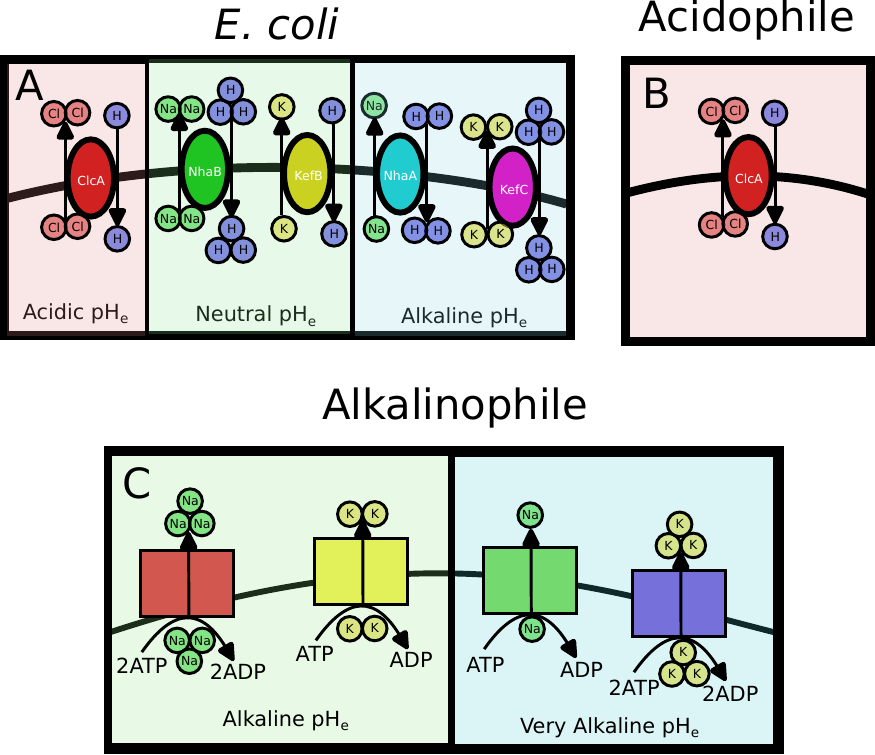}
\caption{\label{fig:5}
Transport architecture reflects the environmental pH range a cell
encounters. (A) Neutrophiles: broad repertoire, PMF-dominated,
low-to-moderate driving forces, KMF weaker than SMF. (B)
Acidophiles: minimal architecture dominated by low-driving-force
AMF systems; specialized for persistently acidic environments.
(C) Alkalinophiles: high-driving-force ATP-coupled cation efflux;
may additionally operate at reduced PMF to mitigate the
maintenance cost, shifting further toward ATP-coupled transport.
In all three regimes, the cost-minimizing architecture follows
directly from the environmental pH range and from the
Section~\ref{sec:2.3} KMF-vs-SMF stoichiometric asymmetry.}
\end{figure}
\section{Discussion}

The central result of this work is that at steady state, the
minimal energetic cost of sustaining any electrochemical gradient
equals the free-energy dissipation imposed by passive ion leakage,
under two assumptions only: steady state, and a thermodynamically
consistent rate law. It holds across transporter architectures,
stoichiometries, energy sources, and organismal identities.

\paragraph{What the framework does.}
This result turns ion homeostasis from a descriptive physiological
phenomenon into a constrained optimization problem. Membrane
permeability, not transporter regulation, sets the baseline
energetic burden that competes with growth and biosynthesis.
Transporter diversity, stoichiometry, and energy coupling do not
abolish this burden; they determine how closely cells can approach
the thermodynamic floor.

\paragraph{Intracellular composition as an optimization output.}
A central implication is that membrane potential and intracellular
ionic composition are partially optimized solutions of a
constrained dissipation-minimization problem. \emph{Permeability
asymmetry is a sufficient condition for the universal K$^+$-rich,
Na$^+$-poor cytoplasm; sodium toxicity is neither required nor
sufficient to explain the pattern.} This provides a quantitative
energetic basis for a long-recognized cross-taxa
observation~\cite{Dibrova2015}.

\paragraph{Environmental shaping of maintenance cost.}
Extracellular pH, salinity, and temperature reshape the energetic
landscape of IMF maintenance in predictable ways:
\begin{center}
\begin{tabular}{lll}
Regime & Dominant leakage & Mechanism\\ \hline
Acidophile & Protons & High $[\text{H}^+]_e$\\
Alkalinophile & Cations & Strongly negative $\Delta\psi$\\
Halophile & Cations & High $[\text{Ion}]_e$\\
Thermophile & All ions & Elevated permeability
\end{tabular}
\end{center}
Extremophily is, thermodynamically, sustained operation closer to
the ceiling of tolerable maintenance power.

\paragraph{Metabolic organization as a consequence of membrane
physics.}
The maintenance burden feeds back on metabolic strategy. Respiratory
architectures require large PMFs; fermentative architectures do not.
\emph{When environmental conditions make sustaining a large PMF
prohibitive, cells face a trade-off between lowering maintenance
dissipation and maintaining pH$_i$~\cite{Terradot2024}.} Under
sufficient stress, metabolic regimes compatible with reduced PMF
become energetically favored. Membrane permeability therefore
constrains not only transport systems but the viable space of
metabolic strategies.

\paragraph{Why the bound is unreachable in practice.}
Real cells necessarily operate above the bound: the catalytic-
energetic efficiency trade-off, finite transporter repertoires,
discrete stoichiometries, and environmental variability each
elevate the cost. Transporter diversity, in this light, is
adaptive modulation of energy quanta and stoichiometry aimed at
approaching the dissipation minimum across heterogeneous
conditions---not redundancy for its own sake.

\paragraph{Testable predictions.}
Several falsifiable predictions follow directly:
\begin{enumerate}
  \item The steady-state maintenance power for IMF preservation
        should scale with measurable ion leakage fluxes along
        environmental gradients.
  \item Organisms adapted to extreme pH or salinity should show
        higher maintenance-energy fractions even when growth is
        limited.
  \item Stress-induced PMF reductions should correlate with shifts
        toward metabolic regimes that tolerate weaker membrane
        polarization.
  \item Potassium transporters should systematically exhibit lower
        H$^+$:K$^+$ or K$^+$:ATP stoichiometry than sodium
        transporters.
\end{enumerate}
Systematic measurements of ion leakage, PMF, and maintenance energy
under controlled perturbations would directly assess how closely
living systems approach the bound. We expect that next-generation
mechanistic growth models will be able to infer these energetic
burdens from simple growth-rate and growth-yield measurements
across environmental conditions.
\medskip

Our modelling framework supports the idea 
that transporter stoichiometry, energy coupling, and metabolic
regime are not independent design features; they are mechanistic
responses to a common thermodynamic constraint imposed by membrane
permeability and the related energetic costs. The modelling 
framework 
developed here expresses this
constraint in a quantitative way. It therefore allows, independently of the specific transport architecture of an organism, 
to estimate the energetic costs of maintaining the
cell in a certain electro-physiological state (essentially internal pH,
and PMF level) in different environments
as characterised by extracellular ionic composition.

\bibliographystyle{unsrt}
\bibliography{biblio}

\appendix
\section{Constants and Parameters}\label{sec:app-constants}

The following tables collect the constant parameters used
for calculations and numerical simulations. Values refer to \textit{E.~coli} at $T = 298$~K
unless otherwise noted.\label{table:constants}

\begin{table}[h!]
\centering
\begin{tabular}{|c|c|c|c|}
\hline
\textbf{Name} & \textbf{Symbol} & \textbf{Units} & \textbf{Value}\\
\hline
Faraday constant & $F$ & C/mol & 96485\\
Gas constant     & $R$ & J/mol/K & 8.31\\
Temperature      & $T$ & K & 298\\
\hline
\end{tabular}
\caption{Physical constants.}
\end{table}

\begin{table}[h!]
\centering
\begin{tabular}{|c|c|c|c|c|}
\hline
\textbf{Name} & \textbf{Symbol} & \textbf{Units} & \textbf{Value} & \textbf{Ref.}\\
\hline
Membrane capacitance & $C_m$ & F\,m$^{-2}$ & $6.5 \times 10^{-3}$ & \cite{bai2006}\\
Cell width  & $w$ & m & $1.07 \times 10^{-6}$ & \cite{basan2015,BudaE5838}\\
Cell length & $l$ & m & $2.95 \times 10^{-6}$ & \cite{basan2015,BudaE5838}\\
\hline
\end{tabular}
\caption{Cell geometry.\label{tableConstants2}}
\end{table}

\begin{table}[h!]
\centering
\begin{tabular}{|c|c|c|c|c|}
\hline
\textbf{Name} & \textbf{Symbol} & \textbf{Units} & \textbf{Value} & \textbf{Ref.}\\
\hline
NADH reduction & $\Delta G_{\text{NADH}}$ & V & $-2.290$ & \cite{flamholz2011}\\
ATP hydrolysis & $\Delta G_{\text{ATP}}$  & V & $-0.560$ & \cite{flamholz2011}\\
\hline
\end{tabular}
\caption{Energy quanta (volts convention).}
\end{table}

\begin{table}[h!]
\centering
\begin{tabular}{|c|c|c|c|c|}
\hline
\textbf{Ion} & \textbf{Symbol} & \textbf{Units} & \textbf{Value} & \textbf{Ref.}\\
\hline
H$^+$  & $P_{\text{H}^+}$  & m/s & $10^{-4}$ & \cite{garlid1989}\footnote{Chosen to match published cation:proton permeability ratios of $10^{-6}$ to $10^{-7}$, given the Na$^+$, K$^+$, Cl$^-$ values below.}\\
OH$^-$ & $P_{\text{OH}^-}$ & m/s & $10^{-4}$ & \footnote{Set equal to $P_{\text{H}^+}$; dissipation due to OH$^-$ leakage is shown negligible in Table~\ref{table:Neglecting_OH}.}\\
Na$^+$ & $P_{\text{Na}^+}$ & m/s & $5.37\times 10^{-10}$ & \cite{costa1989}\\
K$^+$  & $P_{\text{K}^+}$  & m/s & $17.65\times 10^{-10}$ & \cite{costa1989,Robertson1983}\\
Cl$^-$ & $P_{\text{Cl}^-}$ & m/s & $1.38\times 10^{-10}$ & \cite{costa1989}\\
\hline
\end{tabular}
\caption{Membrane permeabilities. The ratio
$P_{\text{H}^+}/P_{\text{cation}}$ is consistent with published
values~\cite{garlid1989}. The relation $P_{\text{K}^+} > P_{\text{Na}^+}$
is as reported in Ref.~\cite{Robertson1983}.}
\end{table}

In addition to the constants above, the model depends on
the total extracellular ion concentration $[\text{Ion}]_e$,
the extracellular pH, the captive charge valence $z_{\text{Y}}$,
and the osmotic pressure difference $\Delta\Pi$, all of which
are varied across simulations.

\begin{table}[h!]
\centering
\begin{tabular}{|c|c|c|c|}
\hline
\textbf{Symbol} & \textbf{Units} & \textbf{Reference value} & \textbf{Range}\\
\hline
$[\text{Ion}]_e \equiv \sum_x [x]_e$ & mol/m$^3$ & 100 (= 100~mM) & 10--1000\\
pH$_e$ & -- & 7 & 3--11\\
$z_{\text{Y}}$ & -- & $-1$ & $-2$ to $+2$\\
$\Delta\Pi$ & atm & 0 & 0--20\\
\hline
\end{tabular}
\caption{Model parameters varied across simulations.}
\end{table}

\section{Cell shape}\label{sec:app-shape}
We assume as in Ref.~\cite{Terradot2024} 
that the cell is a spherocylinder of total length $l$ and width $w$.
This was determined  from microscopy images~\cite{Buda2016} and is consistent with Ref.~\cite{basan2015}.

Surface and volume are then given by:
\begin{align}
    \label{eq:A1}
    S &= \pi w l \\
    \label{eq:A2}
    V &= \pi \left(\dfrac{w}{2}\right)^2 \left(l - \dfrac{w}{3} \right)
\end{align}

In simulations we take $w \approx 1.07~\mu$m and $l \approx 2.95~\mu$m (Table~\ref{tableConstants2}), 
giving an aspect ratio close to the $3\!:\!1$ value observed across several growth conditions~\cite{TAHERIARAGHI2015}. 
When calculating osmotic pressure, we assume that cellular volume stays 
constant and ignore the mechanical properties of the cell wall.

\section{Membrane potential model}
\label{appendix:C}

We start from the single-cation membrane-voltage equation where
the (single) cation is written $\text{C}^+$,
and the anion is written $\text{A}^-$ (as in Ref.~\cite{Terradot2024})
\begin{equation}
\label{eq:C1}
\Delta\psi = \frac{FV}{SC_m}
\left(z_{\text{Y}}[\text{Y}]_i + [\text{H}^+]_i - [\text{OH}^-]_i
+ [\text{C}^+]_i - [\text{A}^-]_i\right)
\end{equation}
and extend it to the three-ion case used throughout this paper:
\begin{equation}
\label{eq:C2}
\Delta\psi = \frac{FV}{SC_m}
\left(z_{\text{Y}}[\text{Y}]_i + [\text{H}^+]_i - [\text{OH}^-]_i
+ [\text{Na}^+]_i + [\text{K}^+]_i - [\text{Cl}^-]_i\right)
\end{equation}
Using fractional extracellular concentrations
$\alpha_x \equiv [x]_e / [\text{Ion}]_e$ and the captive-charge
density $\alpha_{\text{Y}} \equiv z_{\text{Y}}[\text{Y}]_i /
[\text{Ion}]_e$, with $[\text{Ion}]_e \equiv \sum_x [x]_e$, and
defining the reference ratio
\begin{equation}
\label{eq:C6prefactor}
\lambda \;\equiv\; \frac{FV}{SC_m}\, \frac1{[\text{Ion}]_e}
\end{equation}
(which evaluates numerically to $\approx 35$ for our reference
cell geometry and $[\text{Ion}]_e=10$~mM ; see below), we rewrite
Eqs.~\eqref{eq:C1}--\eqref{eq:C2} as
\begin{align}
\label{eq:C6}
\lambda^{-1} \Delta\psi
  &= \alpha_{\text{Y}}
  + \alpha_{\text{C}^+}\, e^{\eta(\Delta G_{\text{C}^+} - \Delta\psi)}
  - \alpha_{\text{A}^-}\, e^{\eta(\Delta G_{\text{A}^-} + \Delta\psi)}\\
\label{eq:C7}
\lambda^{-1} \Delta\psi
  &= \alpha_{\text{Y}}
  + \alpha_{\text{Na}^+}\, e^{\eta(\Delta G_{\text{Na}^+} - \Delta\psi)}
  + \alpha_{\text{K}^+}\, e^{\eta(\Delta G_{\text{K}^+} - \Delta\psi)}
  - \alpha_{\text{Cl}^-}\, e^{\eta(\Delta G_{\text{Cl}^-} + \Delta\psi)}
\end{align}

Where the electrochemical potential (or Ionic motive Forces) was defined as:
\begin{equation}
    \Delta G_x = z_x \Delta \psi + \eta^{-1} \ln \left(\dfrac{[x]_i}{[x]_e}\right)
\end{equation}
where $z_x$ is the valency of ion $x$ and $\eta$ was defined in Eq.~\eqref{eq:flux-generic}.

\paragraph{Approximation used for voltage computation.}
We work with the approximation
\begin{equation}
\label{eq:C8}
\forall x \in \{\text{C}^+,\text{A}^-\}: \quad
0 \;=\; \alpha_{\text{Y}}
+ \sum_x z_x\, \alpha_x\, e^{\eta(\Delta G_x - z_x \Delta\psi)},
\end{equation}
where $\alpha_x$ is further approximated by
$[\text{CA}]_0 / (2[\text{CA}]_0) = 1/2$. The robustness of this
approximation is verified in the next paragraph.

\paragraph{Three regimes of voltage maintenance.}
\begin{description}
  \item[Donnan regime.] Setting all IMFs to zero and varying
        $\alpha_{\text{Y}}$ between $\min(\alpha_{\text{Y}})$ and
        $\max(\alpha_{\text{Y}})$ yields a band
        $[\Delta\psi_0^-,\Delta\psi_0^+]$ of membrane potentials
        achievable at zero active-transport cost:
        Eqs.~(C8A)--(C8C) [labels unchanged].
  \item[Below Donnan ($\Delta\psi < \Delta\psi_0^-$).] Set
        $\alpha_{\text{Y}}=\min(\alpha_{\text{Y}})$ and
        AMF$=0$; solve~\eqref{eq:C8} for
        $\Delta G_{\text{C}^+}$.
  \item[Above Donnan ($\Delta\psi > \Delta\psi_0^+$).] Set
        $\alpha_{\text{Y}}=\max(\alpha_{\text{Y}})$ and
        CMF$=0$; solve~\eqref{eq:C8} for
        $\Delta G_{\text{A}^-}$.
\end{description}

\paragraph{Approximated form of the membrane potential}

to demonstrate that Eq.~\eqref{eq:C8} is a robust approximation of Eq.~\eqref{eq:C1} we first compute the $\alpha_\text{Y}$, $\Delta G_{\text{C}^+}$ and $\Delta G_{\text{A}^-}$ that achieve the desired membrane potential (see previous paragraph) using the approximated form Eq.~\eqref{eq:C8}. We then input the triplet $\{\alpha_\text{Y},\Delta G_{\text{C}^+},\Delta G_{\text{A}^-}\}$ to the non approximated form Eq.~\eqref{eq:C1} and measure the difference in the approximated and non approximated form of the membrane potential, as shown in figure \ref{Appendix_fig:Approximation_DV}. For the non approximated form, ion concentrations are given by \cite{Terradot2024}:
\begin{align}
\forall \text{pH}_e < 7: [\text{A}^-]_e &= \dfrac{[\text{CA}]_0}{1 + 10^{\text{pKb}+\text{pH}_e - \text{pK}_w}} + 10^{\text{pH}_e} - 10^{\text{pH}_e - \text{pK}_w}\\
\forall \text{pH}_e > 7: [\text{C}^+]_e &= \dfrac{[\text{CA}]_0}{1 + 10^{\text{pKa} - \text{pH}_e}} - 10^{-\text{pH}_e} + 10^{\text{pH}_e - \text{pK}_w}
\end{align}
where $[\text{CA}]_0$ is the concentration of salt, \textit{e.g.} NaCl and we pick pKa = -6.3 (that of HCl), pKb = -0.56 (that of NaOH) and pK$_w$ = 14. And the amount of captive charge is given by:
\begin{equation}
z_{\text{Y}} [\text{Y}]_i = \alpha_{\text{Y}} \cdot \left([\text{C}^+]_e + [\text{A}^-]_e \right)
\end{equation}

\begin{figure}[h!]
\begin{center}
\includegraphics[scale=0.18]{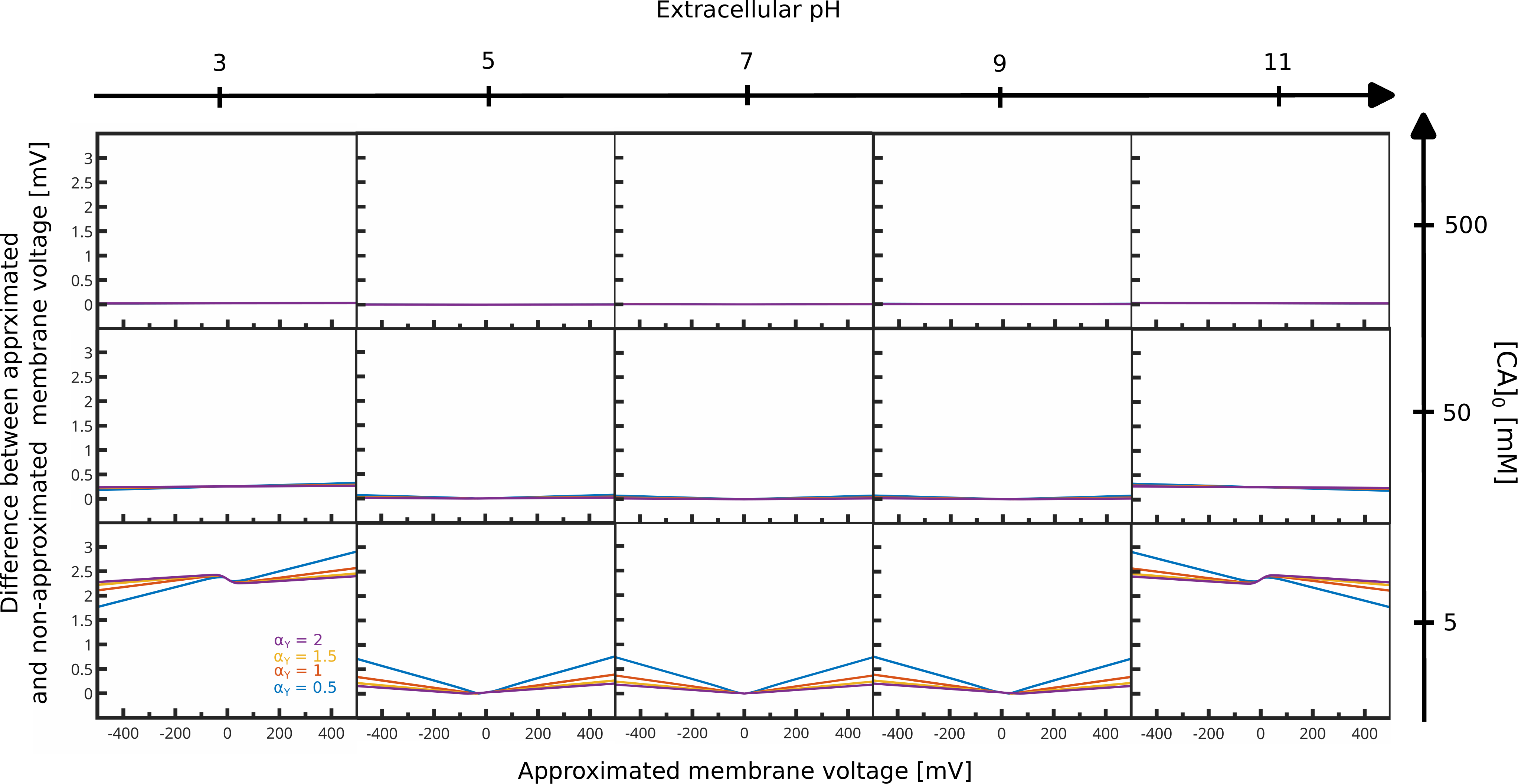}
\caption{\label{Appendix_fig:Approximation_DV} Difference between the approximated and exact membrane potential for identical inputs ${\alpha_\mathrm{Y}, \Delta G_{\mathrm{C}^+}, \Delta G_{\mathrm{A}^-}}$, shown for varying captive charge $\alpha_\mathrm{Y}$ (absolute values indicated in the figure). The membrane potential has the same sign as $\alpha_\mathrm{Y}$.}
\end{center}
\end{figure}

\paragraph{Extending the model to two cations}

We would like to evaluate whether the form Eq.~\eqref{eq:C8} is also a good approximation for the modified set of ions: $x \in \{\text{Na}^+,\text{K}^+,\text{Cl}^-\}$:
\begin{equation}
\label{eq:C9}
0 = \alpha_\text{Y} + \left(\alpha_{\text{Na}^+} \cdot e^{\eta \left(\Delta G_{\text{Na}^+} - \Delta \psi \right)} + \alpha_{\text{K}^+} \cdot e^{\eta \left(\Delta G_{\text{K}^+} - \Delta \psi \right)} -  \alpha_{\text{Cl}^-}  \cdot e^{\eta \left(\Delta G_{\text{Cl}^+} + \Delta \psi \right)} \right)
\end{equation}
We can in fact guarantee that the approximation holds equating:
\begin{align}
\label{eq:C10}
\alpha_{\text{Na}^+} \cdot e^{\eta \left(\Delta G_{\text{Na}^+} - \Delta \psi \right)} + \alpha_{\text{K}^+} \cdot e^{\eta \left(\Delta G_{\text{K}^+} - \Delta \psi \right)} &= 
\alpha_{\text{C}^+} \cdot e^{\eta \left(\Delta G_{\text{C}^+} - \Delta \psi \right)}\\
\label{eq:C11}
\alpha_{\text{Na}^+} + \alpha_{\text{K}^+} &= \alpha_{\text{C}^+} = 1/2
\end{align}
\textit{i.e.} if the approximation holds for a set of values for $\{\alpha_{\text{C}^+},\Delta G_{\text{C}^+},\Delta \psi\}$ on the right hand side, it must hold for the left-hand side since they are allocated the same value. Indeed, for each pair of values the extracellular cationic balance $\alpha_{\text{Na}^+}/\alpha_{\text{K}^+}$ and the KMF (Potassium Motive Force) take, the sodium will take the value such that the left-hand side equate the right-end side, guaranteeing the approximation holds for the set of ions $\{\text{Na}^+,\text{K}^+,\text{Cl}^-\}$. That is, as long as the approximation holds for the right hand side, that depends on the range of CMF (cation motive force) the approximation was sampled for.\\

To guarantee the left hand side is within the same range of value the right hand side took during the evaluation of the approximation we perform the numerical experiment in two steps:
\begin{enumerate}
\item First we solve Eq.~\eqref{eq:C8} in the two dimensions $\alpha_Y \in \{-2,-1,0,+1,+2\}$ and $\Delta G_{\text{C}^+} \in [-500, 0]$ mV for which the approximation was evaluated to be robust in \cite{Terradot2024}.
\item We then rearrange Eq.~\eqref{eq:C10} as to express the SMF (Sodium Motive Force) as a function of the KMF (Potassium Motive Force):
\begin{equation}
\label{eq:C12}
\Delta G_{\text{Na}^+}  =  \eta^{-1}\ln \left(\dfrac{1}{2}\dfrac{1}{\alpha_{\text{Na}^+}} \cdot e^{\eta\left(\Delta G_{\text{C}^+} - \Delta \psi\right)} - \dfrac{\alpha_{\text{K}^+}}{\alpha_{\text{Na}^+}}\cdot e^{\eta\left(\Delta G_{\text{K}^+} - \Delta \psi \right)}\right) + \Delta \psi
\end{equation}
\item and to ensure fine grained sampling of the SMF and KMF we also express the KMF as a function of the SMF:
\begin{equation}
\label{eq:C13}
\Delta G_{\text{K}^+} = \eta^{-1} \ln \left(\dfrac{1}{2}\dfrac{1}{\alpha_{\text{K}^+}} \cdot e^{\eta\left(\Delta G_{\text{C}^+} - \Delta \psi\right)} - \dfrac{\alpha_{\text{Na}^+}}{\alpha_{\text{K}^+}} \cdot e^{\eta\left(\Delta G_{\text{Na}^+} - \Delta \psi\right)}\right) + \Delta \psi
\end{equation}
and introducing the extracellular sodium to extracellular potassium ratio:
\begin{equation}
\beta = \dfrac{\alpha_{\text{Na}^+}}{\alpha_{\text{K}^+}}
\end{equation}
we transform Eqs.~\eqref{eq:C12} and \eqref{eq:C13} into:
\begin{align}
\label{eq:C14}
\Delta G_{\text{Na}^+}  =  \eta^{-1}\ln \left[\left(\dfrac{1}{\beta} +1\right) \cdot e^{\eta\left(\Delta G_{\text{C}^+} - \Delta \psi\right)} - \dfrac{1}{\beta}\cdot e^{\eta \left(\Delta G_{\text{K}^+} - \Delta \psi \right)}\right] + \Delta \psi\\
\label{eq:C15}
\Delta G_{\text{K}^+} = \eta^{-1}\ln \left[\left( 1 + \beta\right)\cdot e^{\eta \left(\Delta G_{\text{C}^+} - \Delta \psi\right)} - \beta \cdot e^{\eta \left(\Delta G_{\text{Na}^+} - \Delta \psi\right)}\right] + \Delta \psi
\end{align}
which we use to get all the combinations of SMF and KMF $\{\Delta G_{\text{Na}^+},\Delta G_{\text{K}^+}\}$ that result from a given pair CMF and membrane potential $\{\Delta G_{\text{C}^+},\Delta \psi\}$ given $\beta$, the ratio of extracellular concentrations of sodium to potassium.
\end{enumerate}

\begin{figure}
\begin{center}
\includegraphics[scale=0.3]{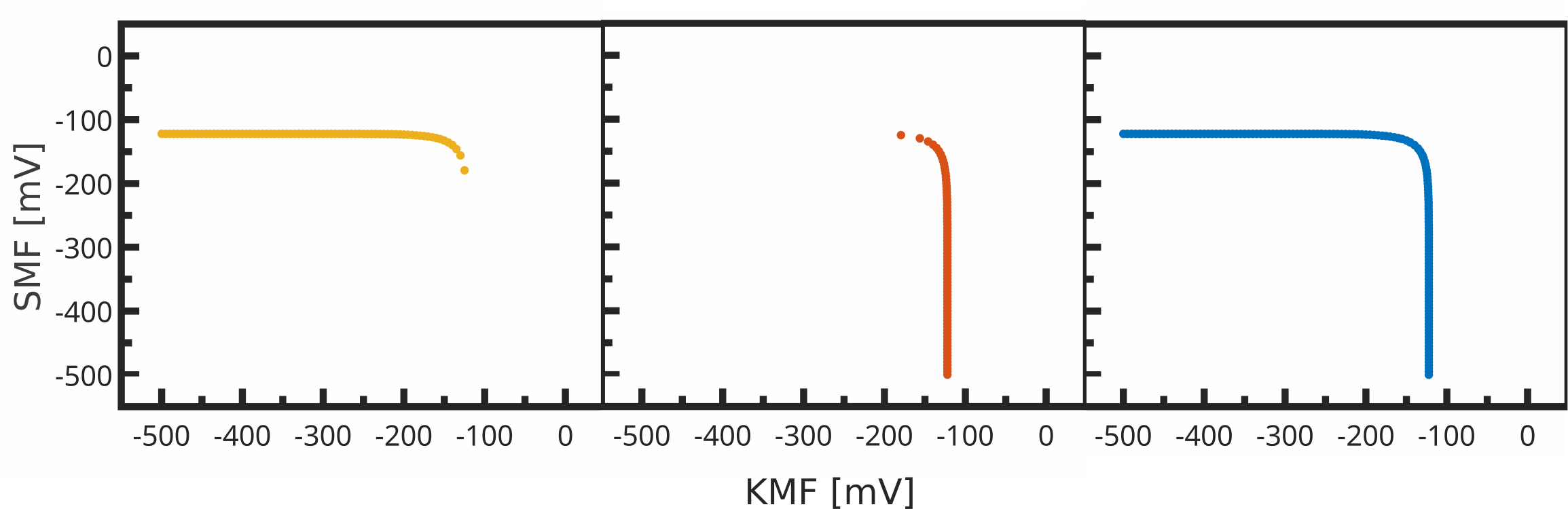}
\caption{\label{Appendix_fig:figure2}
$\beta = 1$, $\Delta \psi = -140$ mV, $\alpha_\mathrm{Y} = 0$. (Left) Sampling over $\Delta G_{\mathrm{Na}^+}$ using Eq.~\eqref{eq:C15}. (Middle) Sampling over $\Delta G_{\mathrm{K}^+}$ using Eq.~\eqref{eq:C14}. (Right) Superposition of the two samplings.}
\end{center}
\end{figure}
\section{Model for leakage}
\label{appendix:D}

We now write the rate at which ions $x$ leak across the membrane
and dissipate their motive force $\Delta G_x$. The generic
thermodynamically consistent rate law~\cite{beard2008} is
\begin{equation}
\label{eq:leak-generic}
j_r \;=\; j_r^{+}\!\left(1 - e^{\eta\, \Delta G_r}\right),
\end{equation}
where the transport mechanism---e.g.\ the number of elementary
steps~\cite{Keener:08:book}---determines the non-negative forward
rate $j_r^{+} \ge 0$. Specializing to the passive leakage of ion
$x$, the reaction free energy equals the ionic motive force
$\Delta G_x$, so
\begin{equation}
\label{eq:leak-ion}
j_x \;=\; j_x^{+}\!\left(1 - e^{\eta\, \Delta G_x}\right).
\end{equation}

\paragraph{Trapezoidal-barrier forward rate.}
We use the trapezoidal energy-profile model of
\cite{garlid1989}---a piecewise-linear approximation to the
potential-energy landscape an ion traverses crossing the
membrane---to write the forward rate:
\begin{equation}
\label{eq:leak-forward}
j_x^{+}(b,u) \;=\; \frac{S}{V}\, P_x\, [x]_e\, f_b(u),
\end{equation}
where $P_x$ is the permeability and
\begin{equation}
\label{eq:leak-udef}
u \;\equiv\; -\, z_x\, \eta\, \Delta\psi.
\end{equation}
The minus sign and the charge factor $z_x$ ensure that $u$ is
positive when the electrical contribution to leakage is outward for
ion $x$, independent of whether $x$ is a cation or an anion. The
voltage factor is
\begin{equation}
\label{eq:leak-fb}
f_b(u) \;=\; -u\, \frac{b\, e^{u/2}}{e^{-bu/2} - e^{bu/2}},
\end{equation}
where $b \in [0,1]$ parametrizes the shape of the electric-potential
drop across the membrane: $b$ characterizes
$d\Delta\psi(z)/dz$ where $z$ is position across the bilayer.

\paragraph{Two limiting cases.}
\begin{itemize}
  \item \emph{$b = 0$ (single-Eyring).} The electric potential
        $\Delta\psi(z)$ is concentrated at a single barrier in the
        middle of the bilayer, so $d\Delta\psi/dz = 0$ everywhere
        except at the geometric center. In this limit
        \eqref{eq:leak-fb} reduces to
        \begin{equation}
        \label{eq:leak-f0}
        f_0(u) \;=\; e^{u/2}.
        \end{equation}
        Note that $f_0(u)$ diverges exponentially as $u$ grows,
        reflecting the single-barrier approximation's breakdown
        at large driving forces---this divergence is the origin
        of the selective-pressure prediction in Section~2.2.
  \item \emph{$b = 1$ (GHK).} The electric potential drops
        linearly across the bilayer. This is the assumption
        underlying the Goldman--Hodgkin--Katz flux equation, and
        \eqref{eq:leak-fb} reduces to
        \begin{equation}
        \label{eq:leak-f1}
        f_1(u) \;=\; -\,\frac{u}{e^{-u} - 1},
        \end{equation}
        which is the dimensionless form of the standard GHK flux
        kernel.
\end{itemize}

\paragraph{Use in the main text.}
Equation~\eqref{eq:leak-forward} is the explicit form of the
forward rate $j_x^{+}$ used in the main text (Eq.~3 of the main
text) and underlies every appearance of $f_b(u)$ in
Sections~2.2--2.7.

\section{Lower Bound on energy expenditure}
\label{appendix:E}

Recall we consider two models:
Model~A (transporters powered directly by
intracellular metabolism) or Model~B (proton-coupled transporters
with metabolic proton extrusion). The following theorem holds for both.
In this theorem, we 
assume as usual: (i) steady state, and (ii)
thermodynamically consistent flux laws
(Eq.~\eqref{eq:leak-generic}). Recall we are working (whichever model we use) with three types of transport reactions:

- leakage reactions $x_e \to x_i$

- coupled transport reactions of which there are two types:
PMF-coupled transport of the form \[H^+_e + x_i \to H^+_i + x_e\]
and metabolism-coupled transport \[E^\star + x_i \to E + x_e\]
and that $Q_T$ is defined as the total amount of energy diverted from intracellular metabolism to maintain a certain electrophysiological state:
\begin{equation}
\label{eq:QT_def}
Q_T \equiv -F \sum_{r \in \text{MCT}} \Delta G_{E,r} j_r
\end{equation}
where the sum is over all metabolically coupled transport reactions and where $\Delta G_{E,r}$ is the amount of energy consumed per turnover of the reaction $r$.

\begin{theorem}[Universal lower bound]\label{thm:bound}
Suppose each metabolism-coupled transport reaction $r$ is such that
$\Delta G_{E,r} \le 0$ (its energy-providing half-reaction is spontaneous).

Then the total power density consumed by metabolism-coupled transport reactions
is lower-bounded:
\begin{equation}
\label{eq:theorem-bound}
Q_T \;\ge\; - F \sum_{x \backslash \{\text{OH}^-\}} \Delta G_x\, j_x \;\ge\; 0
\end{equation}
with equality attained in the limit of $100\%$ energetic
efficiency of all transport reactions.
\end{theorem}

The remainder of this section proves Theorem~\ref{thm:bound} for
each model. The proofs share a single lemma, stated first.

\begin{lemma}[Consistency $\Rightarrow$ per-ion bound]\label{lem:per-ion}
This lemma holds for any ion moved across the membrane by a single active transport reaction (whether PMF-coupled or metabolism-coupled) and by leakage (for model A and B, all ions but protons and hydroxide ions). 
Consider a transport reaction $r$ with $E^\star_r \rightleftharpoons E_r$ consuming $\Delta G_{E,r} \le 0$ amount of energy per reaction turnover, coupled
at stoichiometry $\nu_{x,r}$ to the leakage of ion $x$:
\begin{equation}
\label{eq:W.0b}
r = E^\star + \nu_{x,r}x_i \rightleftharpoons E + \nu_{x,r}x_e
\end{equation}
with steady-state given by
\begin{equation}
\label{eq:W.4}
j_x - \nu_{x,r}j_r = 0
\end{equation} 
The steady-state constraint forces
\begin{equation}
\label{eq:W.4b}
\Delta G_{E,r}/\nu_{x,r} \le \Delta G_x \le 0
\end{equation}

Defining the
energetic efficiency
\begin{equation}
\label{eq:W.1}
\epsilon_{x,r} \equiv \nu_{x,r}\Delta G_x / \Delta G_{E,r}
\end{equation}
we have
$0 \le \epsilon_{x,r} \le 1$, and the per-ion power density (positive)
\begin{equation}
\label{eq:W.2}
Q_x \equiv -F \Delta G_{E,r} j_r
\end{equation}
satisfies
\begin{equation}
\label{eq:W.3}
Q_x \;=\; -\frac{F}{\epsilon_{x,r}}\,\Delta G_x\, j_x \;\ge\;
-F\, \Delta G_x\, j_x
\end{equation}
\end{lemma}

\begin{proof}[Proof of Lemma~\ref{lem:per-ion}]
\emph{Upper bound on $\Delta G_x$.} If $\Delta G_x > 0$, then
$\Delta G_r  < 0$ where
\begin{equation}
\label{eq:W.3b}
\Delta G_r = \Delta G_{E,r} - \nu_{x,r}\Delta G_x
\end{equation}
and by
Eq.~\eqref{eq:leak-generic} $j_x < 0$ and $j_r > 0$, violating the
steady-state constraint~\eqref{eq:W.4}

Hence
$\Delta G_x \le 0$.
\emph{Lower bound on $\Delta G_x$.} Symmetrically, if
$\Delta G_x < \Delta G_{E,r}/\nu_{x,r}$, then $\Delta G_r > 0$, so
$j_x > 0$ and $j_r < 0$, again violating steady state. Hence
$\Delta G_x \ge \Delta G_{E,r}/\nu_{x,r}$. The efficiency bounds
follow immediately from the definition. Finally, substituting the
steady-state condition Eq.~\eqref{eq:W.4} into the definition of $Q_x$ Eq.~\eqref{eq:W.2} yields
\begin{equation}
\label{eq:W.5}
Q_x = -F \Delta G_{E,r} \dfrac{1}{\nu_{x,r}} j_x
\end{equation}
and substituting the stoichiometric ratio above using Eq.~\eqref{eq:W.1} gives:
\begin{equation}
\label{eq:W.6}
Q_x = -F\Delta G_x j_x / \epsilon_{x,r}
\end{equation}, and $\epsilon_{x,r} \le 1$ gives
the stated inequality~\eqref{eq:W.3}.
\end{proof}

\paragraph{Computing the electrochemical potential for hydroxide ions}

we write the steady state for hydroxide ions, including water hydrolysis and OH leak:
\begin{equation}
\label{eq:OH.1}
0 = j_{\text{OH}^-} + f_w [\text{H}_2\text{O}]_i^2 - b_w [\text{H}_3\text{O}^+]_i [\text{OH}^-]_i
\end{equation} 
where if $[\text{H}_2\text{O}]_i = 1000/18 M$ and $b_w$, the rate of association of hydronium and hydroxide ions, is diffusion limited, so that $b_w = 10^{10} \text{M}^{-1} \text{s}^{-1}$, then $f_w = K_w b_w / [\text{H}_2\text{O}]_i^2$ or approximately $3.2 \times 10^{-8}$ M$^{-1}$ s$^{-1}$, assuming $K_w = 10^{-14}$ M$^2$, see \cite{Terradot2024}.
Depending on the model for leakage $[\text{OH}^-]_i$ is obtained by solving
\begin{equation}
\label{eq:OH.2}
\begin{split}
0 = \dfrac{S}{V} P_{\text{OH}^-} [\text{OH}^-]_e e^{\eta \Delta \psi / 2} \left(1 - \exp \left[-\eta \Delta \psi +  \ln \left(\dfrac{[\text{OH}^-]_i}{[\text{OH}^-]_e} \right)\right] \right)\\ + f_w [\text{H}_2\text{O}]_i^2 - b_w [\text{H}_3\text{O}^+]_i [\text{OH}^-]_i
\end{split}
\end{equation}

\begin{equation}
\label{eq:OH.3}
\begin{split}
0 = -\dfrac{S}{V} P_{\text{OH}^-} [\text{OH}^-]_e \dfrac{\eta \Delta \psi}{e^{-\eta \Delta \psi} - 1} \left(1 - \exp \left[-\eta \Delta \psi +  \ln \left(\dfrac{[\text{OH}^-]_i}{[\text{OH}^-]_e} \right)\right] \right)\\ + f_w [\text{H}_2\text{O}]_i^2 - b_w [\text{H}_3\text{O}^+]_i [\text{OH}^-]_i
\end{split}
\end{equation}
where we used Erying model Eq.~\eqref{eq:leak-fb} for the first equation and GHK model Eq.~\eqref{eq:leak-f1} for the second. We find that the electrochemical potential for hydroxide ion $\Delta G_{\text{OH}^-} = z_x \Delta \psi + \eta^{-1} \ln \left(\dfrac{[\text{OH}^-]_i}{[\text{OH}^-]_e} \right)$ is of opposite sign than that of the proton motive force, see figure \ref{fig:sign_DGOH}.

\subsection{Model A: direct metabolic coupling}

Having derived the per-ion power density Eq.~\eqref{eq:W.6} for all ions but protons and hydroxide ions, we next turn to the derivation of the power density for protons, which differs depending on whether we consider Model A (where transport of all ions is metabolically-coupled) or Model B (where only proton transport is metabolically-coupled and all other ion transport is PMF-coupled).

\paragraph{Model's Reactions}

In Model A, we assume that, on top of water hydrolysis, there is for each ion (besides hydroxide ions), a single transport reaction\footnote{assuming two
reactions would give similar results, indeed minimizing energetic expenditures
implies that of two reactions, only one is active given a certain electrochemical
state \cite{Terradot2024}.} that couples a spontaneous half-reaction from internal metabolism $E^\star_r \rightleftharpoons E_r$ (\textit{e.g.} ATP hydrolysis or NADH reduction) to the export of ion $x$
\begin{align}
\label{rec:A.1}
\forall x \in \{\text{H}^+,\text{Na}^+,\text{K}^+,\text{Cl}^-\}: E^\star_r + \nu_{x,r} x_i &\rightleftharpoons E_r + \nu_{x,r} x_e\\
\label{rec:A.2}
2 \text{H}_2\text{O} &\rightleftharpoons \text{H}_3\text{O}^+ + \text{OH}^-
\end{align}

\paragraph{Non-proton ions.}
For every $x \in \{\text{Na}^+,\text{K}^+,\text{Cl}^-\}$ the
transport reaction~\eqref{rec:A.1} and the steady-state
condition~\eqref{eq:W.4} reduce directly to the hypotheses of
Lemma~\ref{lem:per-ion}, yielding
$Q_x \ge -F\, \Delta G_x\, j_x$.

\paragraph{Protons and hydroxide ions.}
For protons, water hydrolysis couples $\text{H}^+$ and
$\text{OH}^-$ steady states which modifies the steady-state equation Eq.~\eqref{eq:W.4} for protons to
\begin{equation}
\label{eq:W.7}
0 = j_{\text{H}^+} - j_{\text{OH}^-} - \nu_{\text{H}^+,r} j_r
\end{equation}
where we substituted the term for water hydrolysis by hydroxide ion leak using Eq.~\eqref{eq:OH.1}.
Repeating the
Lemma~\ref{lem:per-ion} argument on the modified steady-state
equation Eq.~\eqref{eq:W.7}, given $j_{\text{OH}^-}$ is of opposite sign to $j_{\text{H}^+}$ (opposite electrochemical potential, see Fig.~\ref{fig:sign_DGOH}) gives
\begin{equation}
\label{eq:W.8}
\dfrac{\Delta G_{E,r}}{\nu_{\text{H}^+,r}} \leq \Delta G_{\text{H}^+} \leq 0    
\end{equation}
and the energetic efficiency is bounded
\begin{equation}
\label{eq:W.9}
\epsilon_{\text{H}^+,r} \equiv \dfrac{\nu_{\text{H}^+,r} \Delta G_x}{\Delta G_{E,r}} \in [0,1]
\end{equation}
. The power density that reaction $r$ that generates the PMF then follows as:
\begin{equation}
\label{eq:W.10}
Q_{\text{H}^+} = -F \Delta G_{E,r} j_r
\end{equation}
substituting $j_r$ using Eq.~\eqref{eq:W.7}, injecting the energetic efficiency by substituting $\Delta G_{E,r}$ from Eq.~\eqref{eq:W.9} we find
\begin{equation}
\label{eq:W.11}
Q_{\text{H}^+} = -F \Delta G_{E,r} \dfrac{1}{\nu_{\text{H}^+,r}} \left(j_{\text{H}^+} - j_{\text{OH}^-}\right) = -F \Delta G_{\text{H}^+} \dfrac{1}{\epsilon_{\text{H}^+,r}} \left(j_{\text{H}^+} - j_{\text{OH}^-}\right)
\end{equation}
Applying the definition~\eqref{eq:QT_def} to model A, we get:
\begin{equation}
\label{eq:W.11b}
Q_T^A \equiv -F \sum_{r \in \text{MCT}}  \Delta G_{E,r} j_r 
\end{equation}
where we have for each ion a metabolically-coupled transport (MCT) reaction, see reaction~\eqref{rec:A.1}. Going from Eq.~\eqref{eq:W.2} to Eq.~\eqref{eq:W.6} for all non proton ions and going from Eq.~\eqref{eq:W.10} to Eq.~\eqref{eq:W.11} for proton ions, $Q_T^A$ becomes
\begin{equation}
\label{eq:W.11c}
Q_T^A = - F \left(\Delta G_{\text{H}^+} \dfrac{1}{\epsilon_{\text{H}^+,r}} \left(j_{\text{H}^+} - j_{\text{OH}^-}\right) + \sum_{x \backslash \{\text{H}^+, \text{OH}^-\}} \dfrac{1}{\epsilon_{x,r}} \Delta G_x j_x 
\right)
\end{equation}
And the above is bounded by assuming maximal energetic efficiency for all reactions, $\forall x: \epsilon_{x,r} = 1$:
\begin{equation}
\label{eq:A-bound}
Q_T^A \;\ge\; -F\left(-\Delta G_{\text{H}^+}\, j_{\text{OH}^-}
+ \sum_{x\backslash \text{OH}^-} \Delta G_x\, j_x\right)
\end{equation}
Table~\ref{table:Neglecting_OH} shows that the hydroxide term is
always $< 6\%$ of the total, justifying the simplified form used
in the main text,
\begin{equation}
\label{eq:W.12}
Q_T^A \;\ge\; -F \sum_{x\backslash \text{OH}^-} \Delta G_x\, j_x.
\end{equation}

\subsection{Model B: proton-coupled transport}

\paragraph{Model's Reactions}

In Model B, we assume that, on top of water hydrolysis, there is for each ion (besides hydroxide ions and protons), a single transport reaction that couples proton entry $\text{H}^+ \rightleftharpoons \text{H}^+_i$ to the export of ion $x$ and for proton a single spontaneous reaction from internal metabolism $E_M^\star \rightleftharpoons E_M$(\textit{e.g.} ATP hydrolysis or NADH reduction) 
\begin{align}
\label{rec:B.1}
\forall x \in \{\text{Na}^+,\text{K}^+,\text{Cl}^-\}: \nu_{\text{H}^+,r} \text{H}^+_e + \nu_{x,r} x_i &\rightleftharpoons \nu_{\text{H}^+,r} \text{H}^+_i + \nu_{x,r} x_e\\
\label{rec:B.2}
E_M^\star + \nu_{\text{H}^+,M} \text{H}^+_i &\rightleftharpoons E_M + \nu_{\text{H}^+,M} \text{H}^+_e\\
\label{rec:B.3}
2 \text{H}_2\text{O} &\rightleftharpoons \text{H}_3\text{O}^+ + \text{OH}^-
\end{align}
where the potentials of reaction are defined as
\begin{align}
\label{eq:W.13}
\Delta G_r = \nu_{\text{H}^+,r} \Delta G_{\text{H}^+} - \nu_{x,r} \Delta G_x\\
\label{eq:W.14}
\Delta G_M = \Delta G_{E,M} - \nu_{\text{H}^+,M} \Delta G_{\text{H}^+}
\end{align}
and where we identify from the lemma~\ref{lem:per-ion}
\begin{equation}
\label{eq:W.15}
\Delta G_{E,r} \equiv \nu_{\text{H}^+,r} \Delta G_{\text{H}^+}
\end{equation}
and where steady-state for protons reads as
\begin{equation}
    \label{eq:W.16}
 0 =   j_{\text{H}^+} - j_{\text{OH}^-} - \nu_{\text{H}^+,M} j_M + \sum_{x\backslash \{\text{OH}^-,\text{H}^+\}} \nu_{\text{H}^+,r} j_r
\end{equation}
We next show that PMF must be negative. 
\begin{proposition}[PMF must be negative]\label{prop:negative-pmf}
Under Model~B with a spontaneous metabolic proton pump
($\Delta G_{E,M} \le 0$), any steady-state solution satisfies
$\Delta G_{\text{H}^+} \le 0$.
\end{proposition}

\begin{proof}[Proof]
Assume $\Delta G_{\text{H}^+} > 0$. Then assuming spontaneous energy providing half-reaction $E^\star_M \rightleftharpoons E_M$, $\Delta G_{E,M} < 0$ and from Eq.~\eqref{eq:W.14} $\Delta G_M < 0$. Then for any non-proton
antiporter, $\Delta G_x > 0$ is forced (else by
Eq.~\eqref{eq:leak-generic} $j_x > 0$ and $j_r < 0$, violating the
steady-state constraint Eq.~\eqref{eq:W.4}); because $\Delta G_x > 0$ then $j_x < 0$ and steady state requires $j_r < 0$, meaning that $\Delta G_r > 0$. Because $\Delta G_{\text{OH}^-}$ is of opposite sign than $\Delta G_{\text{H}^+}$ then all four flux contributions
in Eq.~\eqref{eq:W.16}---proton leak, hydroxide leak, metabolic
pump $M$, and proton-antiporter terms---export protons out of
the cell. No steady state can then satisfy the proton
balance Eq.~\eqref{eq:W.16}. Hence $\Delta G_{\text{H}^+} \le 0$.
\end{proof}

\paragraph{Non-proton ions.}

With Proposition~\ref{prop:negative-pmf} in hand,
Lemma~\ref{lem:per-ion} applies to each proton-coupled
reaction~\eqref{rec:B.1}, indeed
\begin{equation}
    \label{eq:W.17}
    Q_x = - F \nu_{\text{H}^+,r} \Delta G_{\text{H}^+} j_r
\end{equation}
and
\begin{equation}
\label{eq:W.18}
    \epsilon_{x,r} = \dfrac{\nu_{x,r} \Delta G_x}{\nu_{\text{H}^+,r} \Delta G_{\text{H}^+}}
\end{equation}
where the inequality $0 \leq \epsilon_r \leq 1$ derives from lemma~\ref{lem:per-ion}, having identified Eq.~\eqref{eq:W.15}:
\begin{equation}
\label{eq:W.18b}
\dfrac{\nu_{\text{H}^+,r} \Delta G_{\text{H}^+}}{\nu_{x,r}} \leq \Delta G_x \leq 0
\end{equation}
Further injecting Eq.~\eqref{eq:W.18} in Eq.~\eqref{eq:W.17}, and substituting $j_r$ using the steady-state condition for ion $x$ Eq.~\eqref{eq:W.4}
\begin{equation}
\label{eq:W.19}
    Q_x = -F \nu_{\text{H}^+,r} \Delta G_{\text{H}^+} \dfrac{j_x}{\nu_{x,r}} = -F \dfrac{1}{\epsilon_{x,r}} \Delta G_x j_x
\end{equation}
yielding
\begin{equation}
\label{eq:W.19a}
Q_x \ge -F\,\Delta G_x\, j_x
\end{equation}
for every non-proton ion.

\paragraph{Protons and hydroxide ions.}

We know from proposition~\ref{prop:negative-pmf} that $\Delta G_{\text{H}^+} \leq 0$, and from inequality~\eqref{eq:W.18b} that $\Delta G_r = \nu_{\text{H}^+,r} \Delta G_{\text{H}^+}  - \nu_{x,r} \Delta G_x \leq 0$. In addition $\Delta G_{\text{OH}^-}$ is of opposite sign than $\Delta G_{\text{H}^+}$, see figure~\ref{fig:sign_DGOH}. In order for steady-state Eq.~\eqref{eq:W.16} to be achievable therefore requires negative $\Delta G_M$, from Eq.~\eqref{eq:W.14}, this means:
\begin{equation}
    \label{eq:W.19b}
    \dfrac{\Delta G_{E,M}}{\nu_{\text{H}^+,M}}\leq \Delta G_{\text{H}^+} \leq 0
\end{equation}
and the energetic efficiency
\begin{equation}
\label{eq:W.19c}
    \epsilon_M \equiv \dfrac{\nu_{\text{H}^+,M} \Delta G_{\text{H}^+}}{\Delta G_{E,M}}
\end{equation}
is bounded $\in [0,1]$. We then isolate $j_M$ from Eq.~\eqref{eq:W.16}:
\begin{equation}
    \label{eq:W.20}
    j_M = \dfrac{1}{\nu_{\text{H}^+,M}} \left(j_{\text{H}^+} - j_{\text{OH}^-} + \sum_r \nu_{\text{H}^+,r} j_r \right)
\end{equation}
we then substitute $j_r$ using Eq.~\eqref{eq:W.4} to obtain
\begin{equation}
    \label{eq:W.21}
    j_M = \dfrac{1}{\nu_{\text{H}^+,M}} \left(j_{\text{H}^+} - j_{\text{OH}^-} + \sum_r \dfrac{\nu_{\text{H}^+,r}}{\nu_{x,r}} j_x \right)
\end{equation}
before injecting the energetic efficiency Eq.~\eqref{eq:W.18}:
\begin{equation}
        \label{eq:W.22}
    j_M = \dfrac{1}{\nu_{\text{H}^+,M}} \left(j_{\text{H}^+} - j_{\text{OH}^-} + \sum_r \dfrac{\Delta G_x}{\Delta G_{\text{H}^+}} \dfrac{1}{\epsilon_{x,r}} j_x \right)
\end{equation}
We then substitute $\nu_{\text{H}^+,M}$ above, injecting the energetic efficiency of reaction $M$ Eq.~\eqref{eq:W.19c} to obtain
\begin{equation}
        \label{eq:W.23}
    j_M = \dfrac{1}{\epsilon_M} \dfrac{\Delta G_{\text{H}^+}}{\Delta G_{E,M}} \left(j_{\text{H}^+} - j_{\text{OH}^-} + \sum_r \dfrac{\Delta G_x}{\Delta G_{\text{H}^+}} \dfrac{1}{\epsilon_{x,r}} j_x \right)
\end{equation}
Finally applying the definition of $Q_T$ Eq.~\eqref{eq:QT_def} to model B (where the only metabolically-coupled reaction is $M$, see \eqref{rec:B.2}):
\begin{equation}
\label{eq:W.24}
    Q_T^B \equiv -F \Delta G_{E,M} j_M = -F \dfrac{1}{\epsilon_M} \left[(\Delta G_{\text{H}^+} \left(j_{\text{H}^+} - j_{\text{OH}^-} \right) + \sum_{x \backslash \{\text{H}^+,\text{OH}^-\}} \dfrac{1}{\epsilon_{x,r}} \Delta G_x j_x\right]
\end{equation}
which can be decomposed in 
\begin{equation}
    \label{eq:W.25}
    Q_T^B = Q_{\text{H}^+} + \dfrac{1}{\epsilon_M} \sum_{x \backslash \{\text{H}^+,\text{OH}^-\}} Q_x
\end{equation}
where $Q_x$ was defined in Eq.~\eqref{eq:W.19} and where $Q_{\text{H}^+} = -F \dfrac{1}{\epsilon_M} \Delta G_{\text{H}^+} \left(j_{\text{H}^+} - j_{\text{OH}^-}\right)$.
The lower bound then follows from maximizing energetic efficiencies:
\begin{equation}
\label{eq:B-bound}
Q_T^B \;\ge\; -F\left(-\Delta G_{\text{H}^+}\, j_{\text{OH}^-}
+ \sum_{x\backslash \text{OH}^-} \Delta G_x\, j_x\right).
\end{equation}
Table~\ref{table:Neglecting_OH} shows that the hydroxide term is
always $< 6\%$ of the total, justifying the simplified form used
in the main text,
\begin{equation}
\label{eq:W.26}
Q_T^B \;\ge\; -F \sum_{x\backslash \text{OH}^-} \Delta G_x\, j_x.
\end{equation}
Theorem~\ref{thm:bound} follows.

\begin{figure}[h!]
\begin{center}
\includegraphics[scale=0.3]{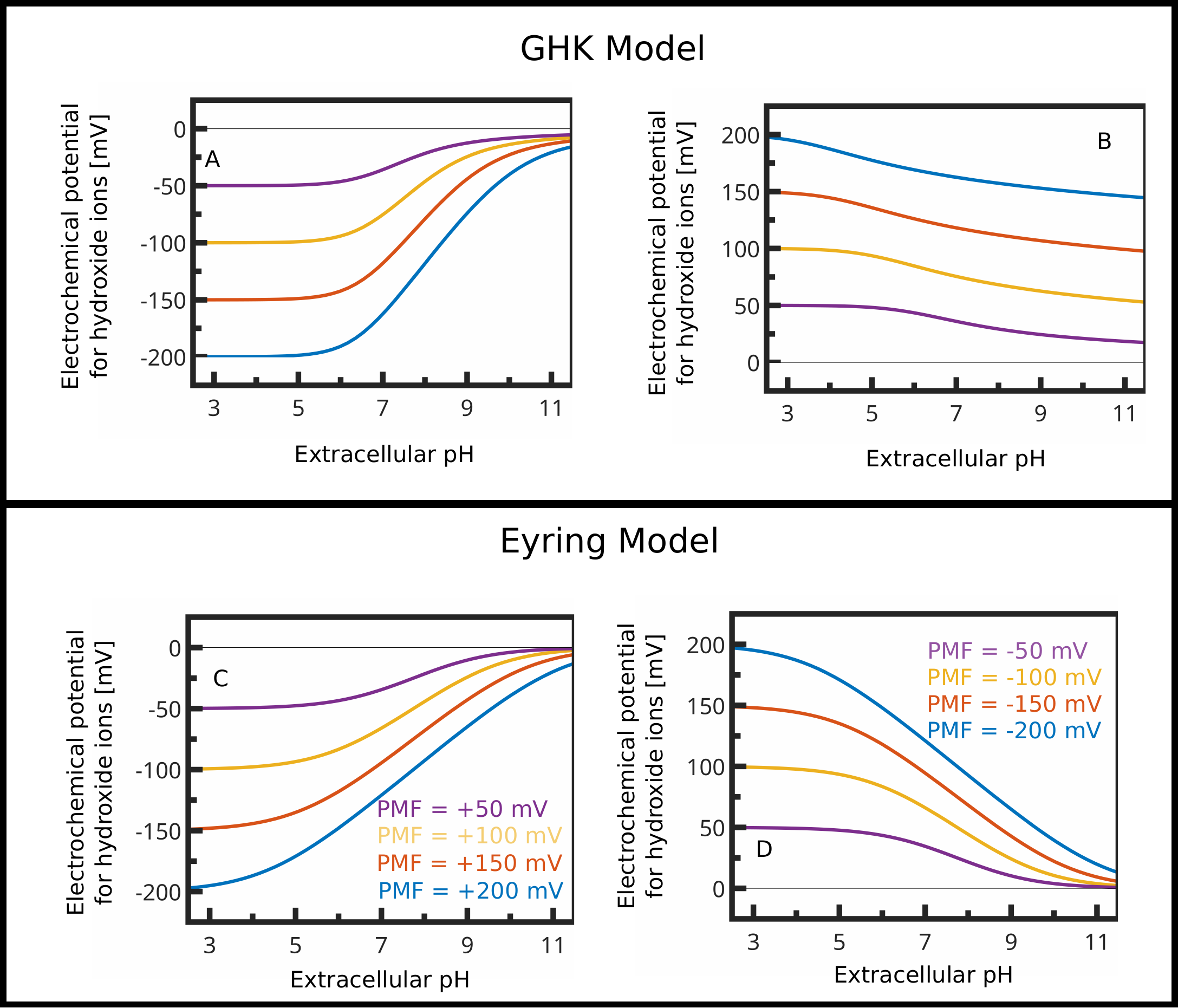}
\caption{\label{fig:sign_DGOH} Electrochemical potential of hydroxide ions, $\Delta G_{\mathrm{OH}^-}$, as a function of extracellular pH for different PMF values. The black line marks $\Delta G_{\mathrm{OH}^-} = 0$. Negative PMF yields positive $\Delta G_{\mathrm{OH}^-}$, whereas positive PMF yields negative $\Delta G_{\mathrm{OH}^-}$. (A, C) Positive PMF values ${+50, +100, +150, +200}$ mV. (B, D) Negative PMF values ${-50, -100, -150, -200}$ mV. Captive charges are set to enforce a fixed osmotic pressure $\Delta \Pi = 1$ atm, with average valency constrained between $-2$ and $+2$.}
\end{center}
\end{figure}

\begin{table}[h!]
\centering
\begin{tabular}{|c|c|c|c|c|}
\hline
pH$_e$ & $\Delta G_{\text{H}^+}$ & $[\text{Ion}]_e$ & leak Model &Fraction of dissipation imputable to $\text{OH}^-$\\
\hline
 3 & -200 mV & 100 mM & Erying & $8.37 \times 10^{-3} \%$ \\
 \hline
  7 &  -200 mV & 100 mM & Erying & $0.70 \%$ \\
 \hline
  11 &  -200 mV &  100 mM & Erying & $3.61 \times 10^{-3} \%$ \\
 \hline
  3 &  -200 mV &  1000 mM & Erying & $8.36 \times 10^{-3} \%$\\
 \hline
  7 &  -200 mV &  1000 mM & Erying & $8.24 \times 10^{-2} \%$\\
 \hline
  11 &  -200 mV &  1000 mM & Erying & $3.61 \times 10^{-4} \%$\\
   \hline
  5 &  -100 mV &  100 mM & Erying & $0.77 \%$ \\
 \hline
  7 &  -100 mV &  100 mM & Erying & $3.98 \%$ \\
 \hline
  9 &  -100 mV &  100 mM & Erying & $0.26 \%$\\
     \hline
  5 &  -100 mV &  1000 mM & Erying & $0.76 \%$ \\
 \hline
  7 &  -100 mV &  1000 mM & Erying & $0.49 \%$ \\
 \hline
  9 &  -100 mV &  1000 mM & Erying & $2.56 \times 10^{-2} \%$\\
 \hline
  3 & -200 mV & 100 mM & GHK & $8.50 \times 10^{-3} \%$\\
 \hline
  7 &  -200 mV & 100 mM & GHK & $3.48 \%$ \\
 \hline
  11 &  -200 mV &  100 mM & GHK &  $0.91 \%$ \\
 \hline
  3 &  -200 mV &  1000 mM & GHK & $8.49 \times 10^{-3} \%$\\
 \hline
  7 &  -200 mV &  1000 mM & GHK & $0.42\%$\\
 \hline
  11 &  -200 mV &  1000 mM & GHK & $9.76 \times 10^{-2} \%$\\
   \hline
  5 &  -100 mV &  100 mM & GHK & $0.77 \%$\\
 \hline
  7 &  -100 mV &  100 mM & GHK &$5.79 \%$ \\
 \hline
  9 &  -100 mV &  100 mM & GHK & $1.73 \%$\\
     \hline
  5 &  -100 mV &  1000 mM & GHK & $0.77 \%$\\
 \hline
  7 &  -100 mV &  1000 mM & GHK &$0.72 \%$ \\
 \hline
  9 &  -100 mV &  1000 mM & GHK & $0.18 \%$\\
 \hline
\end{tabular}
\caption{\label{table:Neglecting_OH} Neglecting hydroxide ion leakage. Fraction of energetic dissipation imputable to ion leake is always smaller than $6\%$, for all the conditions sampled in the main text.}
\end{table}
\section{Dimensional analysis of the leakage power density}
\label{appendix:F}

We verify that the quantity
\begin{equation}
\label{eq:F.1}
C_x \;=\; \Delta G_x\, \frac{S}{V}\, P_x\, [x]_e\, f_b(u)\,
\left(1 - e^{\eta\, \Delta G_x}\right)
\end{equation}
has the intended units. The product of dimensions is
\[
\underbrace{\text{L}^{-1}}_{S/V}\,
\underbrace{\text{L}\,\text{T}^{-1}}_{P_x}\,
\underbrace{\text{mol}\,\text{L}^{-3}}_{[x]_e}\,
\underbrace{\text{J}\,\text{C}^{-1}}_{\Delta G_x}
\;=\; \text{mol}\,\text{T}^{-1}\,\text{L}^{-3}\,\text{J}\,\text{C}^{-1},
\]
so $C_x$ is not yet a power density. Multiplying by the Faraday
constant $F$ ($\text{C}\,\text{mol}^{-1}$) gives
\begin{equation}
\label{eq:F.7}
Q_x \;\equiv\; F\, C_x,
\qquad [Q_x] \;=\; \text{J}\,\text{T}^{-1}\,\text{L}^{-3},
\end{equation}
the intended power density (watt per cubic meter). We use $Q_x$
consistently throughout.

\section{Reference power density of \textit{E.\ coli}}
\label{appendix:G}

The respiratory reaction driving proton antiporters and F$_1$F$_o$
is
\begin{equation}
\label{eq:H.1}
\tfrac{1}{2}\,\text{O}_2 + \text{NADH} + \text{H}^+
\;\rightleftharpoons\; \text{H}_2\text{O} + \text{NAD}^+,
\end{equation}
with free energy $\Delta G_R = 221.3$~kJ/mol. A reference
\textit{E.\ coli} cell consumes $3.08 \times 10^9$ O$_2$ molecules
per cell per hour~\cite{Terradot2024}. Converting to power density,
\begin{equation}
\label{eq:H.2}
Q_0 \;=\;
\frac{3.08 \times 10^9}{N_A}\,\Delta G_R\,\frac{1}{V}\,\frac{1}{3600},
\end{equation}
where the final factor converts per-hour to per-second.
Unit check: (molecules$\cdot$cell$^{-1}$$\cdot$h$^{-1}$) /
(molecules$\cdot$mol$^{-1}$) $\times$ (J$\cdot$mol$^{-1}$) /
(m$^3$$\cdot$cell$^{-1}$) / (s$\cdot$h$^{-1}$) $=$
J$\cdot$s$^{-1}$$\cdot$m$^{-3}$, as required. Numerically,
\begin{equation}
\label{eq:H.8}
Q_0 \;\approx\; 1.35 \times 10^5~\text{W/m}^3,
\qquad Q_0 \cdot V \;\approx\; 3.14 \times 10^{-13}~\text{W/cell}.
\end{equation}

\section{Setting the osmotic pressure}
\label{appendix:H}

To constrain the model to a prescribed osmotic pressure, we write
$\Delta\Pi$ in terms of the ionic imbalance across the membrane:
\begin{align}
\label{eq:I.1}
\frac{\Delta\Pi}{RT} &=
[\text{Y}]_i + [\text{H}^+]_i + [\text{OH}^-]_i
- [\text{H}^+]_e - [\text{OH}^-]_e
\nonumber\\
&
- [\text{C}^+]_e\,\bigl(1 - e^{\eta(\Delta G_{\text{C}^+} - \Delta\psi)}\bigr)
- [\text{Cl}^-]_e\,\bigl(1 - e^{\eta(\Delta G_{\text{Cl}^-} + \Delta\psi)}\bigr)
\end{align}
In the pH$_e \in [3,11]$ range, proton and hydroxide concentrations
are at most $10^{-3}$~M, whereas $[\text{Y}]_i$ and cation terms
are of order $10^{-1}$~M, so the H$^+$/OH$^-$ terms contribute
$< 1\%$ and we drop them:
\begin{equation}
\label{eq:I.2}
\frac{\Delta\Pi}{RT} \;\approx\;
[\text{Y}]_i
- [\text{C}^+]_e\,\bigl(1 - e^{\eta(\Delta G_{\text{C}^+} - \Delta\psi)}\bigr)
- [\text{Cl}^-]_e\,\bigl(1 - e^{\eta(\Delta G_{\text{Cl}^-} + \Delta\psi)}\bigr)
\end{equation}
Substituting~\eqref{eq:I.2} into the approximated voltage
equation~\eqref{eq:C8} eliminates $\alpha_{\text{Y}}$ in favor of
$\Delta\Pi$, yielding (after using
$[\text{C}^+]_e/[\text{Ion}]_e = [\text{Cl}^-]_e/[\text{Ion}]_e =
1/2$) a single equation in
$(z_{\text{Y}}, \Delta G_{\text{C}^+}, \Delta G_{\text{Cl}^-},
\Delta\psi, \Delta\Pi)$:
\begin{equation}
\label{eq:I.4b}
\begin{split}
0 \;=\; z_{\text{Y}}
\left(\frac{\Delta\Pi}{RT} +
\frac{1}{2}\bigl(1 - e^{\eta(\Delta G_{\text{C}^+} - \Delta\psi)}\bigr) +
\frac{1}{2}\bigl(1 - e^{\eta(\Delta G_{\text{Cl}^-} + \Delta\psi)}\bigr)\right)\\
+ \frac{1}{2}\left[e^{\eta(\Delta G_{\text{C}^+} - \Delta\psi)} -
e^{\eta(\Delta G_{\text{Cl}^-} + \Delta\psi)}\right].
\end{split}
\end{equation}

\paragraph{Dispatch.}
Following the same three-case dispatch as in Section~3 (Donnan band,
below Donnan, above Donnan), we solve for
$(z_{\text{Y}},\Delta G_{\text{C}^+},\Delta G_{\text{Cl}^-})$ using
Eq.~\eqref{eq:I.4b} in place of Eq.~\eqref{eq:C8}. The resulting
three Donnan regimes are given by Eqs.~(I.5)--(I.7)

\pagebreak

\section{Supplementary figures}

\begin{figure}[h!]
\centering
\includegraphics[scale=0.20]{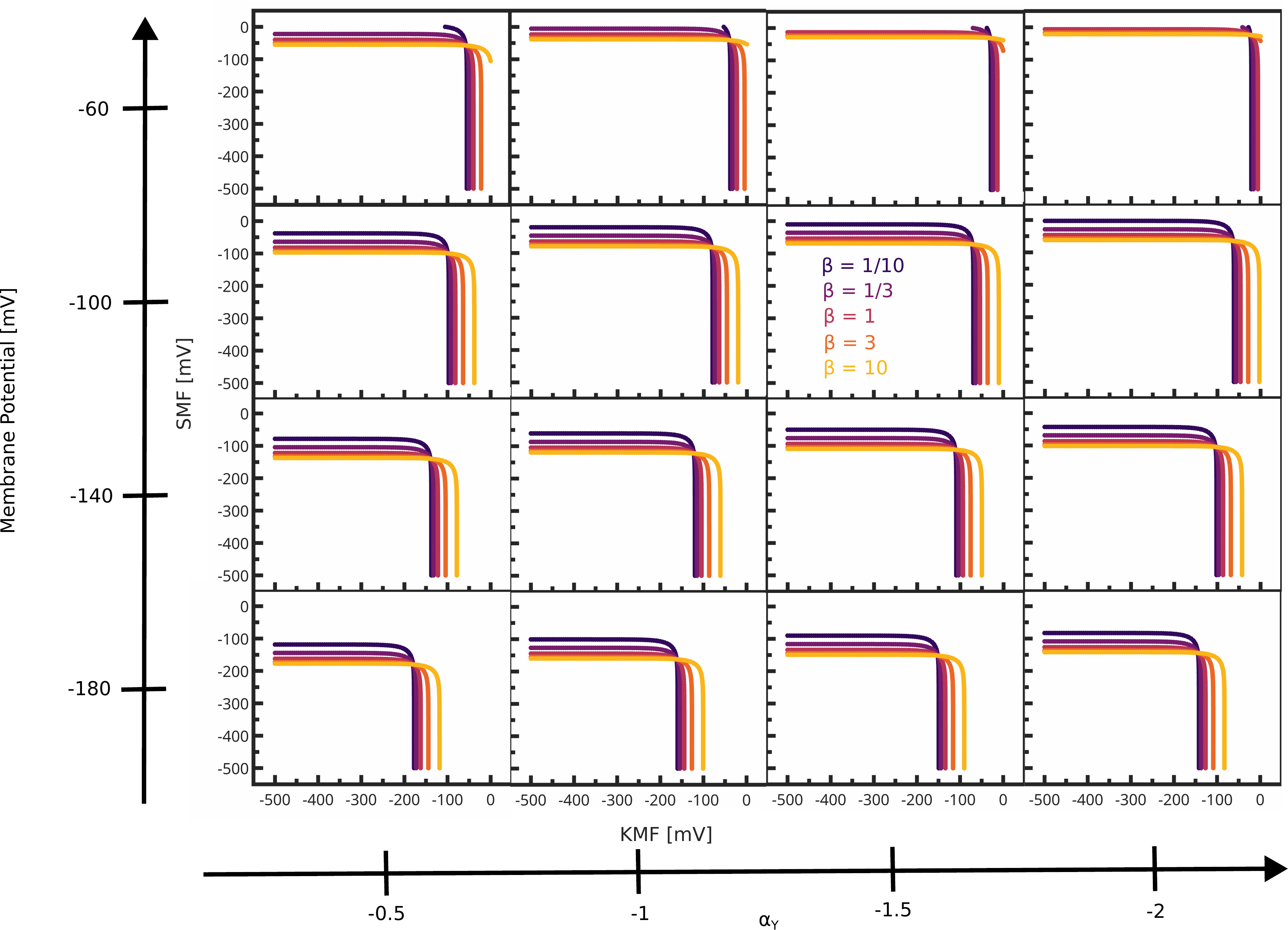}
\caption{\label{Appendix_fig:figure5}
Possible strategies for sustaining negative membrane potentials. From left to right, increasing captive charge ($\alpha_\mathrm{Y}$) reduces the magnitude of the cationic motive forces (CMFs) required to maintain a given membrane potential. From top to bottom, increasingly negative membrane potentials require more negative CMFs, for both potassium (KMF) and sodium (SMF). Insets: different extracellular Na$^+$/K$^+$
 ratios yield distinct CMF pairs compatible with the same membrane potential at fixed $\alpha_\mathrm{Y}$. In general, higher extracellular Na$^+$
 relative to K$^+$
 shifts the optimum toward a less negative KMF, whereas higher extracellular K$^+$shifts it toward a less negative SMF.}
\end{figure}

\begin{figure}[h!]
\centering
\includegraphics[scale=0.4]{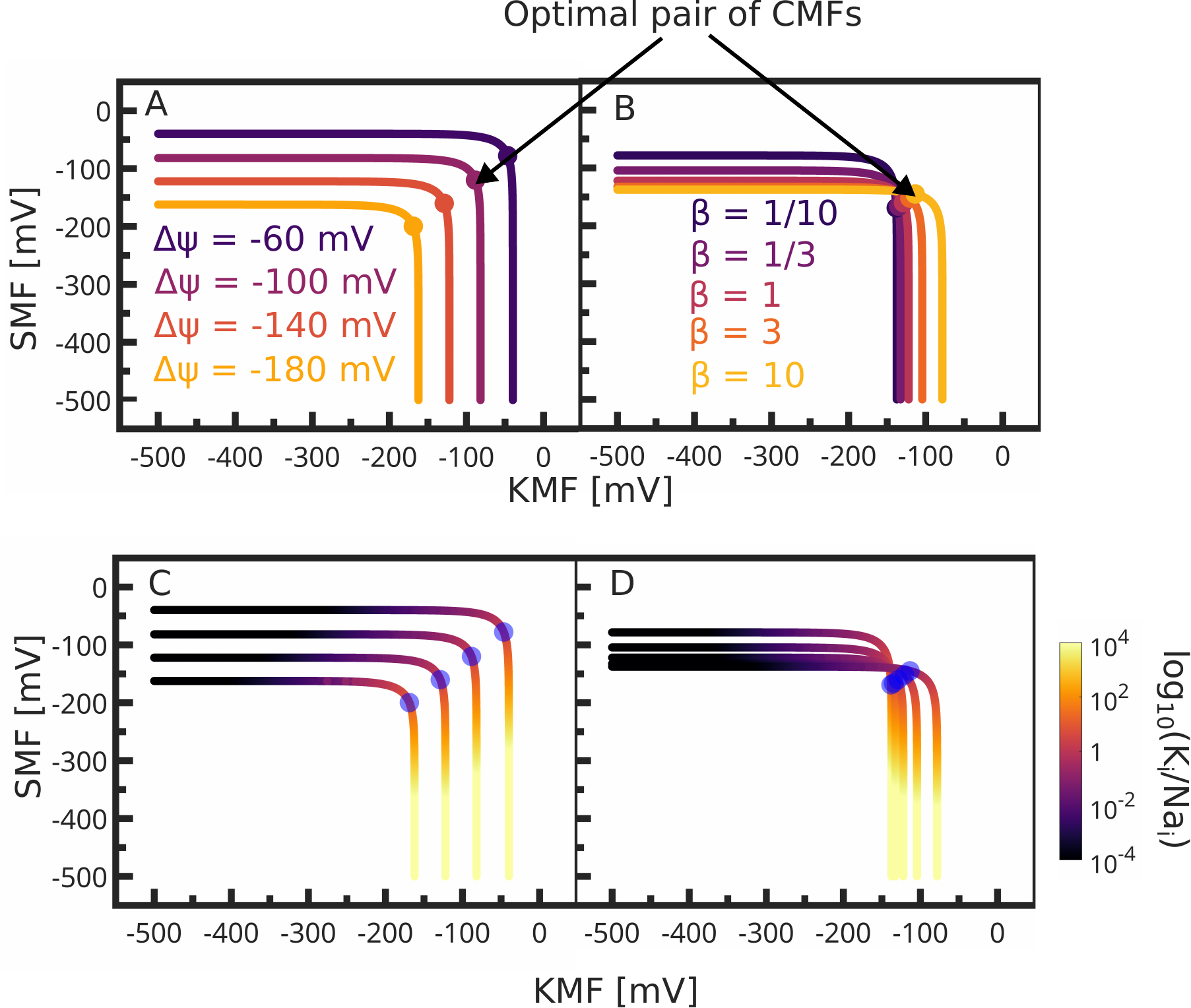}
\caption{\label{Appendix_fig:figure6}
Among all admissible CMF pairs ${\mathrm{KMF}, \mathrm{SMF}}$, a unique pair minimizes the dissipation required to maintain a given membrane potential. Panels (A,B) use $\alpha_\mathrm{Y} = -0.5$ and a permeability ratio $P_{\mathrm{K}^+}/P_{\mathrm{Na}^+} = 10$. (A) At fixed $\beta = 1$, the optimal SMF is systematically more negative than the KMF, reflecting the higher permeability of K$^+$
. (B) At fixed $\Delta\psi = -140$ mV, increasing extracellular Na$^+$
 relative to K$^+$
 (increasing $\beta$) shifts the optimal CMF pair toward more negative values of both KMF and SMF. (C) The resulting intracellular ratios $[\mathrm{K}^+]_i/[\mathrm{Na}^+]_i$ are greater than 1. (D) These ratios vary with $\beta$.}
\end{figure}

\begin{figure}[h!]
\centering
\includegraphics[scale=0.2]{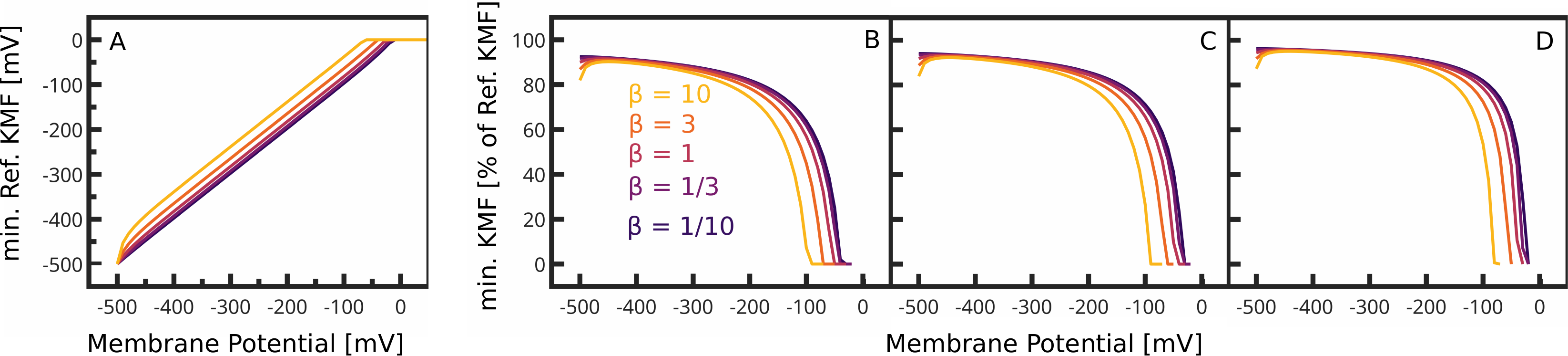}
\caption{\label{Appendix_fig:figure7}
Maximal (least negative) KMF required to sustain a given membrane potential. (A) $\alpha_\mathrm{Y} = -0.5$. (B–D) Ratio of the maximal KMF relative to panel (A), for $\alpha_\mathrm{Y} = -1$, $-1.5$, and $-2$, respectively. Curves correspond to different extracellular ratios $\beta = [\mathrm{Na}^+]_e/[\mathrm{K}^+]e \in {1/10, 1/3, 1, 3, 10}$. Increasing $\beta$ (lower extracellular K$^+$
) shifts the required KMF toward less negative values. Similarly, more negative $\alpha\mathrm{Y}$ (B–D) increases the minimal KMF (closer to 0).}
\end{figure}

\begin{figure}[h!]
\centering
\includegraphics[scale=0.5]{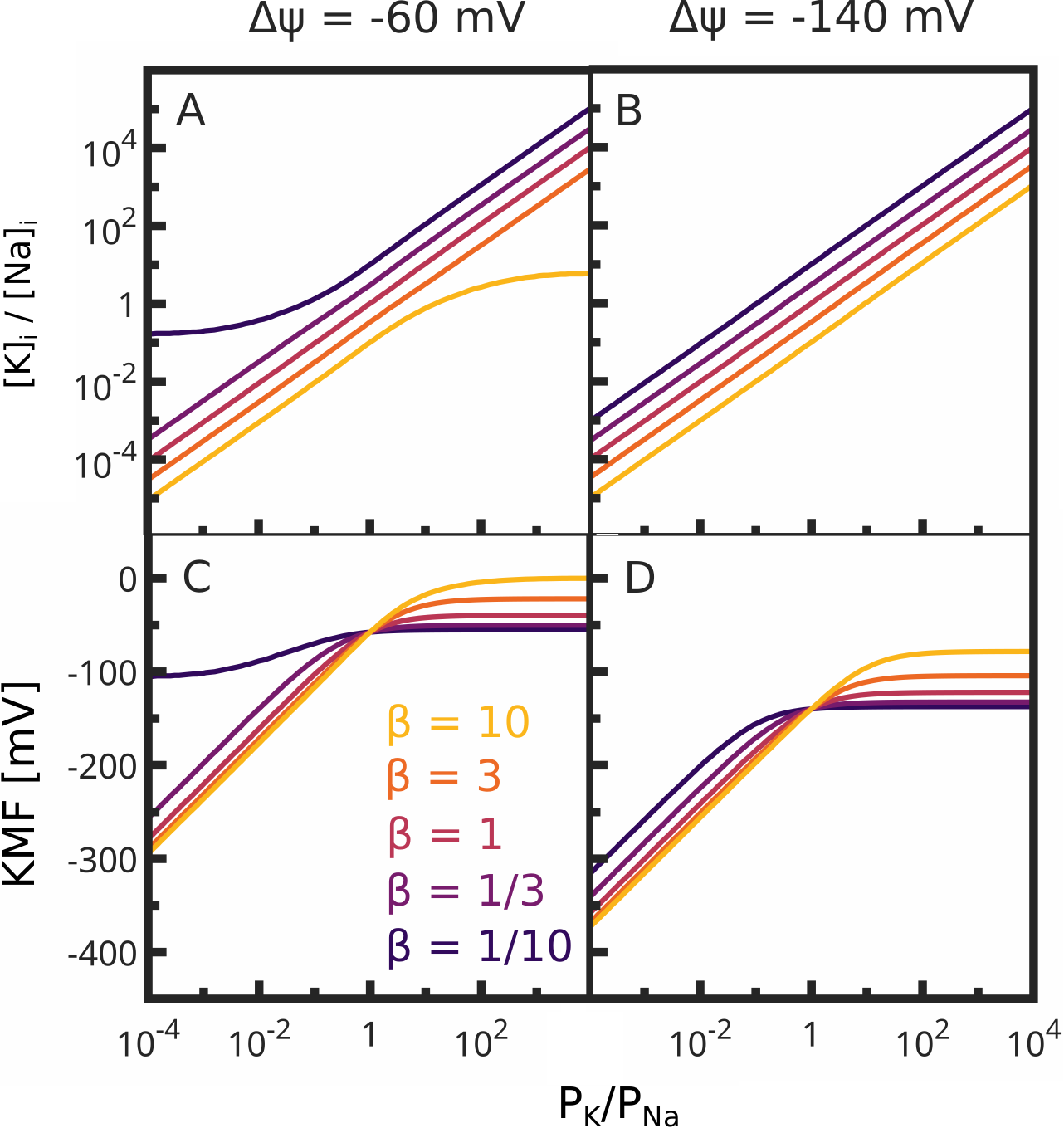}
\caption{\label{Appendix_fig:figure8}
Impact of permeability on optimal strategies for maintaining membrane potential. (A,C) $\Delta\psi = -60$ mV; (B,D) $\Delta\psi = -140$ mV. (A,B) Increasing potassium permeability shifts the optimal intracellular composition toward a higher $[\mathrm{Na}^+]_i/[\mathrm{K}^+]_i$ ratio. (C,D) Increasing potassium permeability shifts the optimal strategy toward a less negative KMF (closer to 0).}
\end{figure}

\begin{figure}[h!]
\centering
\includegraphics[scale=0.5]{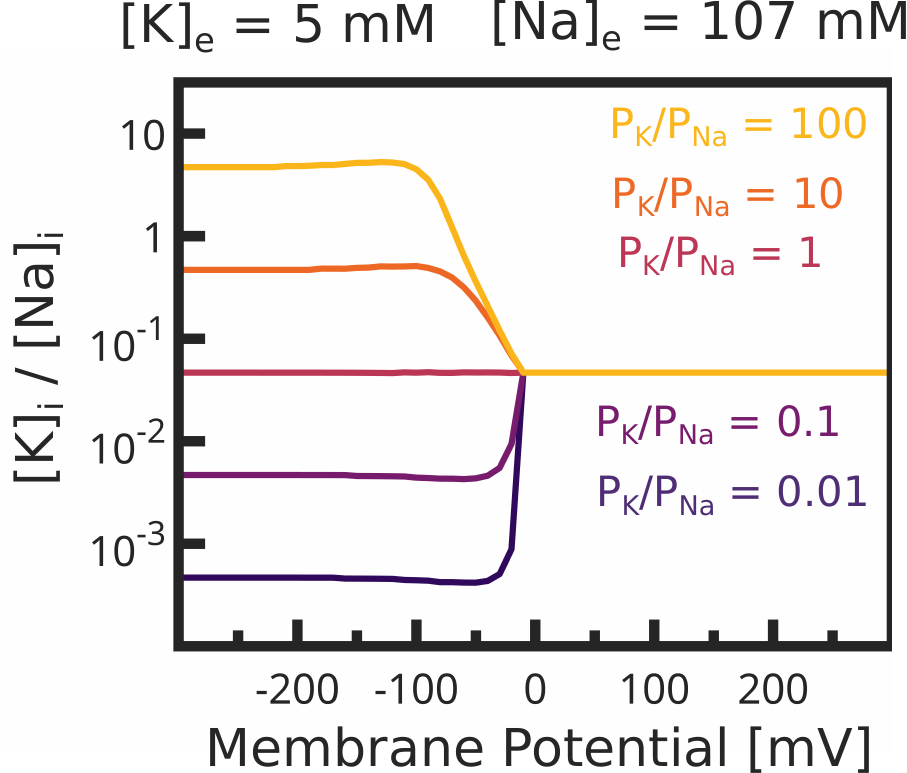}
\caption{
\label{Appendix_fig:figure9}
Impact of the permeability ratio $\gamma_{\mathrm{K}^+} = P_{\mathrm{K}^+}/P_{\mathrm{Na}^+}$ on intracellular cation balance as a function of membrane potential. Higher $\gamma_{\mathrm{K}^+}$ favors optimal strategies with increased intracellular $[\mathrm{K}^+]_i/[\mathrm{Na}^+]_i$. Extracellular composition is taken from \cite{Schultz1961}.} 
\end{figure}

\begin{figure}[h!]
\centering
\includegraphics[scale=0.3]{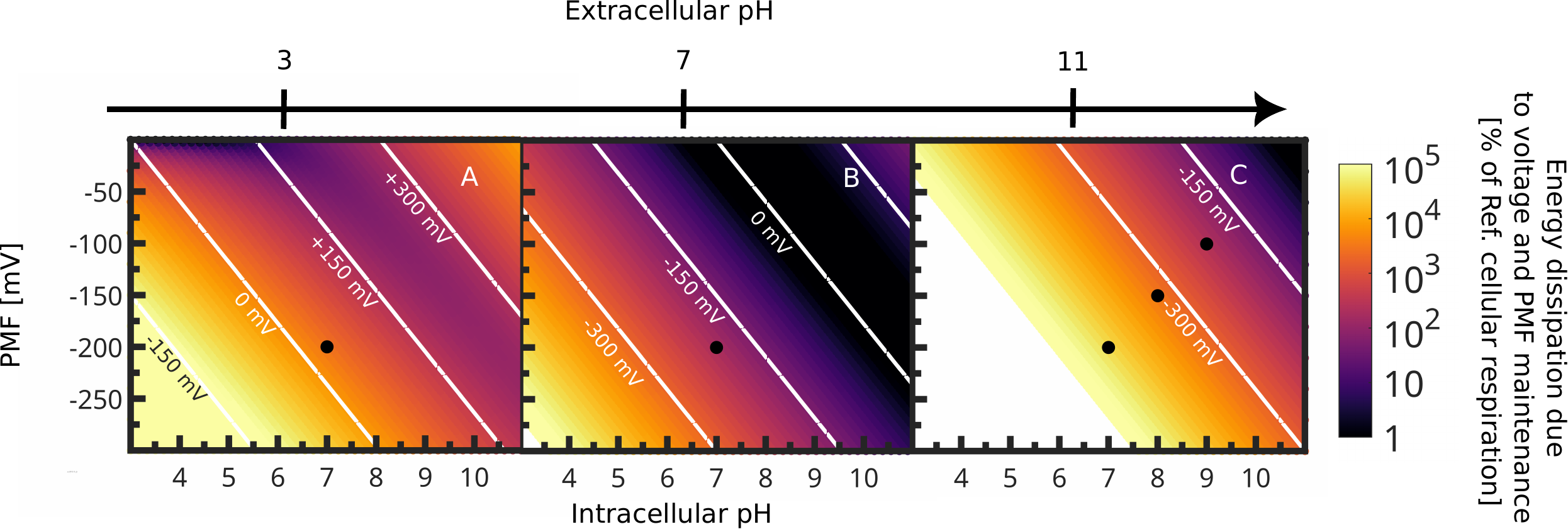}
\caption{\label{Appendix_fig:figure10}
Cost of maintaining intracellular pH and PMF at different extracellular pH: (A) 4, (B) 7, and (C) 11. White lines indicate the corresponding membrane potential. Black dots in panels (A,B) mark the reference state (neutral intracellular pH, PMF = -200 mV). In panel (C), black dots indicate that under alkaline conditions the cost-minimizing state shifts toward reduced PMF (closer to 0) and a more alkaline intracellular pH.} 
\end{figure}

\end{document}